\documentclass[lettersize,onecolumn]{IEEEtran}
\usepackage{cite}
\usepackage{graphicx}
\usepackage{verbatim}
\usepackage{float}
\usepackage{stfloats}
\usepackage{url}
\usepackage{bm}
\usepackage{array}
\usepackage[caption=false,font=normalsize,labelfont=sf,textfont=sf]{subfig}
\usepackage{textcomp}
\usepackage{amsmath}
\usepackage{amsthm} 
\usepackage{amsxtra}
\usepackage{mathrsfs} 
\usepackage{amssymb} 
\usepackage{amsfonts} 
\usepackage{algorithm}
\usepackage{algorithmic}
\usepackage{setspace}
\usepackage{multirow}
\usepackage[T1]{fontenc}
\usepackage{graphicx}
\newtheorem{definition}{Definition}
\newtheorem{theorem}{Theorem}

\newtheorem{lemma}{Lemma}
 
\begin{document}
	
	\title{A Construction of Evolving $k$-threshold Secret Sharing Scheme over A Polynomial Ring}
	
	\author{Qi Cheng, Hongru Cao, Sian-Jheng Lin,~\IEEEmembership{Member,~IEEE}, and Nenghai Yu,~\IEEEmembership{Member,~IEEE}
	\thanks{This work was supported by the National Natural Science Foundation of China under Grant 62071446. (Corresponding author: Sian-Jheng Lin.)}
  \thanks{Qi Cheng, Hongru Cao, Sian-Jheng Lin and Nenghai Yu are with the CAS Key Laboratory of Electromagnetic Space Information, School of Cyber Science and Technology, University of Science and Technology of China, Hefei 230027, China (e-mail: \{cxiaoq,chrkeith\}@mail.ustc.edu.cn; \{sjlin,ynh\}@ustc.edu.cn).}
}	
	\maketitle
	
	\begin{abstract}
		The threshold secret sharing scheme allows the dealer to distribute the share to every participant such that the secret is correctly recovered from a certain amount of shares. The traditional $(k, n)$-threshold secret sharing scheme requests that the number of participants $n$ is known in advance. In contrast, the evolving secret sharing scheme allows that $n$ can be uncertain and even ever-growing. In this paper, we consider the evolving secret sharing scenario. Using the prefix codes and the properties of the polynomial ring, we propose a brand-new construction of evolving $k$-threshold secret sharing scheme for an $\ell$-bit secret over a polynomial ring, with correctness and perfect security. The proposed schemes establish the connection between prefix codes and the evolving schemes for $k\geq2$, and are also first evolving $k$-threshold secret sharing schemes by generalizing Shamir's scheme onto a polynomial ring. Specifically, the proposal also provides an unified mathematical decryption for prior evolving $2$-threshold secret sharing schemes. Besides, the analysis of the proposed schemes show that the size of the $t$-th share is $(k-1)(\ell_t-1)+\ell$ bits, where $\ell_t$ denotes the length of a binary prefix code of encoding integer $t$. In particular, when $\delta$ code is chosen as the prefix code, the share size achieves $(k-1)\lfloor\lg t\rfloor+2(k-1)\lfloor\lg ({\lfloor\lg t\rfloor+1}) \rfloor+\ell$, which improves the prior best result $(k-1)\lg t+6k^4\ell\lg{\lg t}\cdot\lg{\lg {\lg t}}+ 7k^4\ell\lg k$, where $\lg$ denotes the binary logarithm. When $k=2$, the proposed scheme also achieves the minimal share size for single-bit secret, which is the same as the best known scheme.
	\end{abstract}
	
	
	\section{Introduction}
	\IEEEPARstart{T}{he} $(k, n)$ secret sharing scheme encodes the secret to $n$ shares such that the secret can be losslessly recovered from any $k$ out of $n$ shares, and any $k-1$ shares cannot decode any information for the secret. Shamir~\cite{shamir1979share} and Blakley~\cite{blakley1979safeguarding} independently proposed the $(k,n)$ secret sharing scheme in 1979. Based on Shamir's scheme, a lot of schemes~\cite{fuyou2014randomized,harn2016realizing,harn2010strong,pang2005new,yang2004t} has been proposed. However, the conventional secret sharing schemes require that the dealer shall know the maximal number of participants $n$ in advance. When we allow that the dealer can produce more transparencies later, the maximum of $n$ cannot be determined, and hence most conventional schemes cannot be applied to this scenario.
	
	To solve this issue, Komargodski et al.~\cite{komargodski2016share,komargodski2017share} introduced the evolving $k$-threshold secret sharing scheme. In this scheme, the secret can be recovered from any $k$ out of infinitely many participants, and any $k-1$ shares out of these shares cannot get any information about the secret. Specifically, Komargodski et al.~\cite{komargodski2016share,komargodski2017share} first proposed a construction of evolving $2$-threshold secret sharing scheme based on prefix codes. For an $\ell$-bit secret, it shows that the share size of the $t$-th share is no more than $\lg t+(\ell+1)\lg{\lg t}+4\ell+1$ bits. And the proposed scheme is optimal for $1$-bit secret. Futhermore, Komargodski et al.~\cite{komargodski2016share,komargodski2017share} proposed evolving $k$-threshold secret sharing scheme for arbitrary $k$. The $t$-th share size of the proposed evolving $k$-threshold scheme is $(k-1)\lg t+6k^4\ell\lg{\lg t}\cdot\lg{\lg {\lg t}}+ 7k^4\ell\lg k$.
	
	However, firstly, it is unknown whether there are connections between prefix codes and the evolving $k$-threshold secret sharing schemes when $k>2$. This is left as an open problem in \cite{komargodski2017share}. Secondly, the evolving schemes in~\cite{komargodski2017share}, which are constructed based on the idea of distributing shares on a generational basis, and use an evolving scheme once and Shamir's scheme multiple times, become increasingly complex and even not easy to construct as $k$ increases. The natural idea~\cite{komargodski2017share} for constructing the evolving scheme is based on simulating Shamir's scheme. Komargodski et al~\cite{komargodski2017share}, attempted it but failed. How to design algebraic-oriented constructions is also left as an open problem. Thirdly, the corresponding share size in~\cite{komargodski2017share} is not optimal. D’Arco et al~\cite{DARCO2021149} proposed a new evolving $3$-threshold scheme based on the Chinese Remainder Theorem in order to reduce the share size. They~\cite{DARCO2021149} made the share size of the $t$-th participant close to $\lg t+poly(k)\cdot o(\lg t)$ where $k=3$. However, the scheme can only be used in the case of $k=3$. As for $k\geq4$, there are no related works, to the best of our knowledge.
	%The construction scheme ~\cite{d2021secret}, which achieves a smaller share size, can only be used in the case of $k=3$, 
	%D’Arco et al~\cite{d2021secret} have proposed a new technology to construct the evolving $3$-threshold scheme to achieve a smaller share size. However, the technology can only be applied in the case $k=3$.
	%To the best of our knowledge, when $k\geq4$, there only exists one construction of evolving $k$-threshold scheme~\cite{komargodski2017share} at present.
	
	%On the other hand, Komargodski et al. attempted the natural idea~\cite{komargodski2017share} of constructing evolving scheme by simulating Shamir's scheme, but failed.
	%constructing evolving scheme by simulating Shamir's scheme is a natural idea~\cite{komargodski2017share}.  attempted but failed.
	%Instead they constructed the evolving scheme based on the idea of distributing shares on a generational basis.
	%Then, the schemes proposed in ~\cite{komargodski2017share} is based on the idea of distributing shares on a generational basis. 
	%However, the evolving $k$-threshold secret sharing scheme~\cite{komargodski2017share} becomes increasingly complex and even not easy to construct, as $k$ increasing. The authors in ~\cite{komargodski2017share} also leaved how to design algebraic-oriented constructions as an open problem.
	
	To this end, we propose a brand-new construction of evolving $k$-threshold secret sharing scheme for an $\ell$-bit secret over a polynomial ring based on prefix codes, where the size of the $t$-th shares is determined by the codeword size of encoding $t$. The proposed schemes establish the connection between prefix codes and the evolving schemes for $k\geq2$, and achieve algebraic-oriented constructions by generalizing Shamir's scheme. Then we prove the correctness and security of this scheme, and analyze the corresponding share size. Specifically, we first propose the construction of evolving $2$-threshold secret sharing scheme over $F_2[x]$. We also apply the construction to other complicated binary prefix codes to compare the share size with the scheme~\cite{komargodski2017share}. It shows that the proposed scheme can achieve a smaller size for some binary prefix coding. Second, we extend the scheme to the evolving $k$-threshold secret sharing scheme on $F_2[x]$ for $k\geq3$. The analysis of share size shows that the $t$-th share of the proposed scheme is $(k-1)(\ell_t-1)+\ell$ bits, where $\ell_t$ denotes the length of a prefix code of encoding integer $t$. Specifically, using $\delta$ code~\cite{elias1975universal} as the prefix code, the share size is given by $(k-1)\lfloor\lg t\rfloor+2(k-1)\lfloor\lg ({\lfloor\lg t\rfloor+1}) \rfloor+\ell$, which is smaller than the prior result $(k-1)\lg t+6k^4\ell\lg{\lg t}\cdot\lg{\lg {\lg t}}+ 7k^4\ell\lg k$ in~\cite{komargodski2017share}. Third, considering the secret $s\in\{0,1,\cdots,p-1\}^{\ell}$, we proposed a construction of evolving $k$-threshold scheme over the polynomial ring $F_p[x]$, where $k\geq2$. This can be seen as an extension of the proposed scheme. In addition, we show some $p$-ary prefix codings for positive integers and then apply the construction to these codes.
	
	\subsection{Our Contributions}
	The main contributions of
	this paper are enumerated as follows. 
	\begin{itemize}
		\item[$\bullet$] The proposed schemes successfully establish the connection between prefix codes and evolving $k$-threshold secret sharing for arbitrary $k$, where the size of the $t$-th shares is determined by the codeword size of encoding $t$. It means that there is a corresponding construction of evolving $k$-threshold secret sharing scheme for any given prefix codes. In our view, this result can answer part of the open question raised by~\cite{komargodski2017share}.
		
		\item[$\bullet$] It is well known that the evolving $2$-threshold secret sharing scheme has been studied comprehensively. When $k=2$, the proposed scheme provide an unified mathematical decryption for prior evolving $2$-threshold secret sharing schemes. In addition, the proposed scheme also achieves the minimal share size for single-bit secret, as stated in ~\cite{komargodski2016share,komargodski2017share}.
		
		\item[$\bullet$] The proposed brand-new scheme is more concise. This is also the first evolving $k$-threshold secret sharing scheme by generalizing Shamir's scheme onto a polynomial ring.
	\end{itemize}
	
	\subsection{Related works}
	Shamir~\cite{shamir1979share} and Blakley~\cite{blakley1979safeguarding} independently proposed the $(k,n)$ secret sharing scheme in 1979. The construction of Shamir’s scheme was based on Lagrange interpolation. Blakley's scheme was established by using the property of points in multidimensional space. In Shamir's scheme, the dealer randomly chose a polynomial in $F_p[x]$ of degree less than $k$ such that the constant term of this polynomial is the secret, and $p>n$. Consequently, the $n$ shares, which are the evaluations of the polynomial at $n$ distinct points, are respectively distributed to $n$ participants. Therefore, each share is an element of $F_p$, that can be represented with approximately $\lg n$ bits. Later, Karnin et al.~\cite{karnin1983secret} proved that the share size of Shamir's scheme is optimal when the length of the secret is $\lg n$ bits. And Bogdanov et al.~\cite{bogdanov2016threshold} proved that the share size of Shamir’s secret sharing scheme is optimal for a $1$-bit secret. Due to the simplicity and practicality of Shamir’s scheme, it has been widely used and plays a very important role in modern cryptography. 
	
	Besides, Mignotte~\cite{mignotte1983share} and Asmuth et al.~\cite{asmuth1983modular} respectively gave the new constructions of $(k,n)$-threshold secret sharing scheme based on the Chinese Remainder Theorem for integer rings. Both schemes are similar, however, Asmuth-Bloom's scheme improved security compared with Mignotte's scheme. In Asmuth-Bloom's scheme, $n$ integers that are increasing and pairwise coprime were randomly selected, and the product of any $k$ numbers is greater than the value of the secret. Subsequently, the dealer calculated the remainders that module the secret for $n$ integers separately and sent each participant a share consisting of the corresponding integer and its remainder. Due to the superior computational complexity and efficiency of this scheme, it has become a widely studied and highly regarded secret sharing scheme. Then, many works~\cite{liu2015novel,harn2016realizing,ning2018constructing} utilized the technology of Asmuth-Bloom's scheme to achieve various performances in practical applications. 
	
	The general method to construct secret sharing schemes for any given secret sharing function was presented by Benaloh and Leichter~\cite{benaloh1990generalized}. Pedersen~\cite{pedersen1991non} proposed a more convenient and practical secret sharing method. 
	Another secret sharing scheme was considered~\cite{brickell1989some}. 
	In addition, the secret sharing scheme is widely used in other applications, including threshold cryptography~\cite{gennaro1998simplified} and multiparty computations~\cite{cramer2000general}. In 1985, Chor et al.~\cite{chor1985verifiable}
	proposed the concept of verifiability for the first time and constructed a verifiable secret sharing scheme.
	Simmons~\cite{simmons1990really} described a practical application of the access structure of secret sharing. 
	Moreover, Krawczyk’s
	scheme~\cite{krawczyk1993secret} was constructed to achieve computing efficiency with high probability allowing small statistical errors in security. Ding et al.~\cite{ding2021communication} described a new construction for the communication efficient secret sharing scheme with a small share size to minimize
	the decoding bandwidth. 
	
	Komargodski et al.~\cite{komargodski2016share,komargodski2017share} introduced the evolving $k$-threshold secret sharing scheme. Actually, an evolving scheme~\cite{csirmaz2012line} was presented for a similar scenario before Komargodski et al's scheme. Komargodski et al's scheme also left the possibility of constructing a new scheme for the dynamic threshold access structure. Komargodski et al.~\cite{komargodski2017evolving} gave the construction using the algebraic manipulation detection codes. When $k=2$, except the proposed schemes~\cite{komargodski2016share,komargodski2017share}, some related studies were developed based on the scheme. D’Arco et al.~\cite{d2018equivalence} showed the equivalence between binary prefix codes and evolving $2$-threshold secret sharing schemes. Given a secret with arbitrary length, Okamura et al.~\cite{okamura2020new} first applied the construction of an evolving $2$-threshold secret sharing scheme to more complicated binary codes for positive integers, and analyzed the corresponding share size. Furthermore, based on $D$-ary prefix code for $D\geq 3$, the construction of the evolving $2$-threshold scheme was also proposed~\cite{okamura2020new}.

	\subsection{Organization}
	The rest of the paper is organized as follows. In Section~\ref{sec2}, we review the traditional secret sharing scheme and evolving secret sharing scheme, and show the existing main results of the evolving secret sharing scheme. In Section~\ref{sec3} and  Section~\ref{sec4}, we consider the scenario where the secret $s\in\{0,1\}^{\ell}$. Specifically, in Section~\ref{sec3}, we first propose a new construction of evolving $2$-threshold secret sharing scheme, and elaborate on the security of the scheme. Then the evolving $k$-threshold secret sharing scheme is given in Section~\ref{sec4} for $k\geq 3$. In Section~\ref{sec5}, we provide a construction of evolving $k$-threshold secret sharing scheme considering the secret $s\in\{0,1,\cdots,p-1\}^{\ell}$. Finally, Section~\ref{sec6} concludes the work and discusses the unresolved issues.
	
	\section{Models and Notations}\label{sec2}
	Let $\mathbb{N}^+$ denote the set of positive integers. Let $[n]=\{1,2,\cdots,n\}$ for $n\in \mathbb{N}^+$. Let $|A|$ denote the cardinality of the set $A$. For any $p\in \mathbb{N}^+$, let $F_p$ denote the finite field with $p$ elements. Let $F_p[x]=\{\sum_{j=0}^{N}a_jx^{j}|a_j\in F_p\}$ as the polynomial ring, where $N$ is a finite positive integer. Let $F_p[[x]]:=\{\sum_{j=0}^{\infty}a_jx^{j}|a_j\in F_p\}$. 
	
	\subsection{Secret Sharing Scheme} 
	Let $\mathcal{P}_n=\{P_1,P_2,\cdots,P_n\}$ represent the set of $n$ participants. The power set of $\mathcal{P}_n$ is written as $2^{\mathcal{P}_n}$. The collection $\mathcal{A} \subseteq 2^{\mathcal{P}_n}$ is monotone if for arbitrary $A\in \mathcal{A}$ and $A\subseteq C\in 2^{\mathcal{P}_n}$, it holds that $C\in \mathcal{A}$. The access structure is defined as follows.
	\begin{definition}\label{def1}
		$\mathcal{A} \subseteq 2^{\mathcal{P}_n}$ is called an access structure if $\mathcal{A}$ is a monotone collection of non-empty subsets. The subset in $\mathcal{A}$ is called qualified, and the subset in $2^{\mathcal{P}_n}\setminus\mathcal{A}$ is called unqualified. 
	\end{definition}
	
	\begin{definition}\label{def2}
		For $k,n \in \mathbb{N}^+$ with $k\leq n$, 
		the $(k,n)$-threshold access structure $\mathcal{A}$ is a collection that contains all subsets of size is no less than $k$, i.e
		\begin{equation*}
			\mathcal{A}=\{A\in 2^{\mathcal{P}_n}||A|\geq k\}.	
		\end{equation*}
	\end{definition}
	
	A $(k,n)$-threshold secret sharing scheme requires a secret $s\in S$, a set of $n$ participants and a $(k,n)$ access structure $\mathcal{A}$, where $S$ is the domain of the secret. In the scheme, the dealer distributes the share to every participant such that the shares of any subset in $\mathcal{A}$ can correctly recover the secret, while the shares of any subset not in $\mathcal{A}$ cannot gain any information about the secret. 
	
	We denote by ${Z}^{(s)}_i$ the share of the $i$-th participant, and $B({Z}^{(s)}_i)$ represent the bit length of ${Z}^{(s)}_i$ for $1\leq i\leq n$. Generally, a $(k,n)$-threshold secret sharing scheme consists of a pair of algorithms $(\mathcal{E},\mathcal{R})$, where $\mathcal{E}$ is used to encode the secret $s$ into shares, and $\mathcal{R}$ is used to reconstruct the secret from a subset of shares $B\in 2^{\mathcal{P}_n}$. The following requirements shall be satisfied.
	\begin{itemize}
		\item[$\bullet$] Correctness: For any qualified set $A\in \mathcal{A}$, the algorithm $\mathcal{R}$ can correctly recover $s$ from the shares of the participants in $A$, that is
		\begin{equation}
			P[\mathcal{R}(\{Z^{(s)}_i\}_{P_i\in A},A)=s]=1.
		\end{equation}
		\item[$\bullet$] Secrecy: For any unqualified set $C\in 2^{\mathcal{P}_n}\setminus\mathcal{A}$, there is no information about $s$ leaking to the participants in $C$.
	\end{itemize} 
	In particular, the following conclusion is usually used to verify the security of the secret sharing schemes.
	\begin{lemma}\label{lem1}
		Let $s_0, s_1\in S$ be two different secrets. The scheme is secure if for arbitrary $C\in2^{\mathcal{P}_n}\setminus\mathcal{A}$, the two distributions $(\{Z^{(s_0)}_i\}_{P_i\in C})$ and $(\{Z^{(s_1)}_i\}_{P_i\in C})$ are identical.	
	\end{lemma}
	
	Shamir~\cite{shamir1979share} first proposed a construction for the $(k, n)$-threshold secret sharing scheme. For an $\ell$-bit secret $s$, the scheme uses the polynomial over $F_p$ for $p\geq n$. The share size $Z^{(s)}_i$ satisfies $B(Z^{(s)}_i)\geq \max\{\ell,\lg p\}$ for any $i\in [n]$ in Shamir's scheme.
	
	\subsection{Evolving Secret Sharing Scheme}
	When the maximum of $n$ cannot be determined in advance and can even be infinite, the conventional secret sharing schemes cannot be applied directly. In this case, the evolving secret sharing schemes are developed. We denote $\mathcal{P}=\{P_1,P_2,\cdots\,P_n,\cdots\}$ as the set of participants, where $\mathcal{P}$ is possible infinite. Then, we give some definitions of the evolving secret sharing scheme. 
	
	\begin{definition}\label{def3}
		$\mathcal{A} \subseteq 2^{\mathcal{P}}$ is called an evolving access structure if $\mathcal{A}$ is a monotone collection of non-empty subsets and the collection $\mathcal{A}_t \colon=\mathcal{A}\cap\{P_1,P_2,\cdots,P_t\}$ is an access structure for any $t\in \mathbb{N}^+$.
	\end{definition}
	
	\begin{definition}\label{def4}
		For $k \in \mathbb{N}^+$, 
		the evolving $k$-threshold access structure $\mathcal{A}$ is a collection that contains all subsets of size is no less than $k$, i.e
		\begin{equation*}
			\mathcal{A}=\{A\in 2^{\mathcal{P}}||A|\geq k\}.	
		\end{equation*}
	\end{definition}
	
	Similarly, an evolving $k$-threshold secret sharing scheme requires a secret $s\in S$, a set of participants and an evolving $k$-threshold access structure $\mathcal{A}$. In the scheme, the secret can be recovered from any $k$ out of infinitely many participants, and any $k-1$ shares cannot deduce any information about the secret. An evolving $k$-threshold secret sharing scheme also includes a pair of algorithms $(\mathcal{E},\mathcal{R})$, which satisfy the following requirements.  
	\begin{itemize}
		\item[$\bullet$]Composition: For any $j\in \mathbb{N}^+$, the share $Z^{(s)}_j$ of $j$-th participant is constructed by the previous $j-1$ shares $\{Z^{(s)}_i\}_{i=1}^{j-1}$ using the algorithm $\mathcal{E}$, i.e.  
		\begin{equation}
			Z^{(s)}_i=\mathcal{E}(s,\{Z^{(s)}_i\}_{i=1}^{j-1}).
		\end{equation}
		
		\item[$\bullet$]Correctness: For any $t\in \mathbb{N}^+$, $A\in \mathcal{A}_t$, the algorithm $\mathcal{R}$ can correctly recover the secret $s$ from the shares of the participants in $A$.
		
		\item[$\bullet$]Secrecy: For any $t\in \mathbb{N}^+, C\in 2^{\mathcal{P}_t}\setminus\mathcal{A}_t$, there is no information about $s$ leaking to the participants in $C$.    
	\end{itemize} 
	
	Komargodski et al.~\cite{komargodski2016share,komargodski2017share} first constructed an evolving $k$-threshold secret sharing scheme for an ${\ell}$-bit secret $s$, and obtained two main results, which are described as follows. 
	
	\begin{theorem}
		For $\ell,t\in \mathbb{N}^+$, there exists an evolving $2$-threshold secret sharing scheme, where the size of the $t$-th share satisfies
		\begin{equation}
			B(Z^{(s)}_t)\leq \lg t+(\ell+1)\lg{\lg t}+4\ell+1.
		\end{equation}
	\end{theorem}
	
	\begin{theorem}
		For $\ell,k, t\in \mathbb{N}^+$, there exists an evolving $k$-threshold secret sharing scheme, where the size of the $t$-th share satisfies
		\begin{equation}
			B(Z^{(s)}_t)\leq (k-1)\lg t+6k^4\ell\lg{\lg t}\cdot\lg{\lg {\lg t}}+ 7k^4\ell\lg k.
		\end{equation}
	\end{theorem}

	\section{Evolving $2$-threshold Secret Sharing Scheme over $F_2[x]$}\label{sec3}
	In this section, we propose an evolving $2$-threshold secret sharing scheme for an $\ell$-bit secret $s$. The scheme allows the dealer to distribute the share to every participant such that no less than two shares can recover $s$ and any single share cannot construct $s$.
	
	\subsection{Proposed Scheme}
	Given a set of prefix codes for integers, the codeword of the integer $i$ is denoted as $c_i=(c_{i,0},c_{i,1},\cdots,c_{i,\ell_i-1})$, where $\ell_i$ represents the length of $c_i$ with $i\in \mathbb{N}^+$. The polynomial form of $c_i$ is defined as $y_i=\sum_{j=0}^{\ell_i-1}c_{i,j}x^j\in F_2[x]$. Then the $i$-th share with the algorithm $\mathcal{E}$ is defined as
	\begin{equation}\label{key105}
		Z^{(s)}_i=r_0+sy_i \pmod {x^{\ell_i+\ell-1}},
	\end{equation}
	where $r_0$ is randomly chosen in $F_2[[x]]$. Notably, $s$ in (\ref{key105}) uses the polynomial form, which would be the default form throughout this paper. As the degree of $r_0$ is infinite in \eqref{key105}, we cannot choose an $r_0\in F_2[[x]]$ in practice. However, due to the operation modulo $x^{\ell_i+\ell-1}$ in \eqref{key105}, the dealer only needs to choose the part of $r_0$ with degree less than $\ell_i+\ell-1$. 
	
	For any two participants $P_i$, $P_j\in \mathcal{P}$ with $i<j$, let $L_{i,j}$ denote the maximal integer satisfying 
	$x^{L_{i,j}}\mid (y_i-y_j)$. The algorithm $\mathcal{R}$ finds out that the following equation 
	\begin{equation}\label{key106}
		s(y_i-y_j)=Z^{(s)}_i-Z^{(s)}_j \pmod {x^{\ell_i+\ell-1}}
	\end{equation}
	has a unique solution of $s$ in $F_2[x]/x^\ell$ as 
	\begin{equation}\label{key1061}
		s=\frac{(Z^{(s)}_i-Z^{(s)}_j)/x^{L_{i,j}}}{(y_i-y_j)/x^{L_{i,j}}}\pmod {x^\ell}.
	\end{equation}
	The existence and uniqueness of $s$ will be provided in the proof of correctness in next subsection. 
	
	Notably, when the dealer generates the $i$-th share via \eqref{key105}, where $i$ is sufficient large, the dealer needs to construct the corresponding $r_0$. To solve the issue, we provide an algorithm as follows. First, the dealer distributes secret $s$ to the participants in the group. When a new participant joins the group, the dealer then generates and assigns new shares for $s$ to the new participant. Based on this, we suppose that the group have $j\geq 2$ participants initially. The dealer first randomly chooses a suitable polynomial $r_0$ such that the degree is no less than $\ell_j+\ell-2$. According to the selected $r_0$, the dealer distributes the secret $s$ to each participant among $j$ participants via \eqref{key105}. After completing this distribution, the dealer will discard the corresponding $r_0$ and $s$. Therefore, when a new participant $P_{j+1}$ joins the group, the dealer needs to obtain the value of $r_0$ and reconstruct $s$. The dealer first needs to reconstruct all coefficients of $r_0$ with the degree is no more than $\ell_j+\ell-2$, then the dealer randomly chooses the coefficients of $r_0$ such that the degree from $\ell_j+\ell-1$ adding to $\ell_{j+1}+\ell-2$. Next, we show the algorithm to reconstruct the corresponding coefficients of $r_0$ with the degree less than $\ell_{j}+\ell-1$. For the evolving $2$-threshold secret sharing scheme, using known $\{Z^{(s)}_i\}_{i=j-1}^{j}$ and $\{y_i\}_{i=j-1}^{j}$, the dealer can reconstruct the secret $s$ by calculating \eqref{key1061}, then reconstructs the coefficients of $r_0$ with degree less than $\ell_{j}+\ell-1$ as
	\begin{equation}
		r_0=Z^{(s)}_j+sy_j \pmod {x^{\ell_j+\ell-1}}.
	\end{equation}
	Finally, the dealer randomly chooses $\ell_{j+1}-\ell_{j}$ coefficients such that the degree of $r_0$ is $\ell_{j+1}+\ell-2$, then the corresponding $r_0$ is obtained. When another new participant joins the group, the corresponding $r_0$ can aslo be obtained using the above similar method.
	
	\subsection{Proofs of Correctness and Secrecy}
	In this subsection, we will prove the correctness and secrecy of the proposed scheme. Before that, we first emphasize several lemmas, which provide useful results for the proofs. 
	\begin{lemma}\label{lem2}
		For a finite field $F_p$, let $f(x), g(x), h(x), k(x)\in F_p[x]$ with $f(x)\neq 0$ satisfy the congruence equation	
		\begin{equation}\label{key107}
			f(x)g(x) \equiv f(x)h(x) \pmod {f(x)k(x)},
		\end{equation}
		then we have 
		\begin{equation}
			g(x) \equiv h(x) \pmod {k(x)}. 
		\end{equation}	
	\end{lemma}
	
	\begin{proof}
		According to the definition of congruence equation in \eqref{key107}, there exists $h_1(x)\in F_p[x]$ such that 
		\begin{equation}\label{key109}
			f(x)g(x)-f(x)h(x)=f(x)k(x)h_1(x),	  
		\end{equation}
		since $F_p$ is a finite field and $f(x)\neq 0$ in $F_p[x]$, we can further get
		\begin{equation}\label{key110}
			g(x)-h(x)=k(x)h_1(x).
		\end{equation}
		Therefore, we have 
		\begin{equation}\label{key111}
			g(x) \equiv h(x) \pmod {k(x)}.
		\end{equation}  
	\end{proof}
	
	\begin{lemma}\label{lem3}
		For any $n,k\in \mathbb{N}^+$, let $f(x)=a_{n}x^n+a_{n-1}x^{n-1}+\cdots+a_1 x+a_0\in F_p[x]$. If $a_0\neq 0$, then $f(x)$ is invertible over $F_p[x]/x^k$.
	\end{lemma}
	
	\begin{proof}
		Let $g(x)=x^k$, then we infer that the factor of $g(x)$ must be the form of $x^d$ for any $0\leq d\leq k$. However, since $f(0)=a_0\neq 0$, then $x^d$ is not the factor of $f(x)$. Therefore, we have
		\begin{equation}\label{key112}
			(g(x),f(x))=1.
		\end{equation}
		According to Euclidean algorithm, there exists $u(x), v(x)\in F_p[x]$ such that 
		\begin{equation}\label{key113}
			u(x)f(x)+v(x)x^k=1,
		\end{equation} 
		hence, we further derive
		\begin{equation}\label{key114}
			u(x)f(x)\equiv 1 \pmod {x^k}.
		\end{equation}
		Then $u(x)$ modulo $x^k$ is the inverse of $f(x)$ in $F_p[x]/x^k$.	
	\end{proof}
	
	\begin{lemma}\label{lem4}
		For any $k_1, k_2\in \mathbb{N}^+$ with $k_1\leq k_2$, let $f_1(x), g_1(x), f_2(x), g_2(x)$ be polynomials over $F_p$ satisfying the following congruence equations	
		\begin{equation}\label{key115}
			\begin{aligned}
				\left
				\{
				\begin{array}{c}
					f_1(x)\equiv g_1(x) \pmod {x^{k_1}},\\
					f_2(x)\equiv g_2(x)\pmod {x^{k_2}}.\\
				\end{array} \right.
			\end{aligned}
		\end{equation}
		Then we have 
		\begin{equation}
			f_2(x)-f_1(x)\equiv g_2(x)-g_1(x) \pmod{x^{k_1}}. 
		\end{equation}	
	\end{lemma}
	\begin{proof}
		According to the definition of congruence equations in \eqref{key115}, there exists $h_1(x), h_2(x)\in F_p[x]$ such that 
		\begin{equation}\label{key117}
			f_1(x)- g_1(x)=x^{k_1}h_1(x),	  
		\end{equation}
		and 
		\begin{equation}\label{key118}
			f_2(x)- g_2(x)=x^{k_2}h_2(x)=x^{k_1}h_2(x)x^{k_2-k_1}.	  
		\end{equation}
		Subtracting the equation \eqref{key117} from the equation \eqref{key118}, we can further get
		\begin{equation}\label{key119}
			(f_2(x)-f_1(x))- (g_2(x)-g_1(x))=x^{k_1}(h_2(x)x^{k_2-k_1}-h_1(x)).
		\end{equation}
		From \eqref{key119}, we obtain 
		\begin{equation}\label{key120}
			f_2(x)-f_1(x)\equiv g_2(x)-g_1(x) \pmod {x^{k_1}}.
		\end{equation}  
	\end{proof}
	
	\begin{theorem}\label{thm3}
		For any $\ell,\ell_1,k\in \mathbb{N}^+$, given $f(x), h(x) \in F_p[x]$ with $f(x)\neq 0$, let $\ell_1$ denote the maximal integer such that $x^{{\ell}_1}\mid f(x)$ and $x^{{\ell}_1+1}\ mid f(x)$. Let $g(x)$ be a polynomial over $F_p$ with the degree no more than $\ell-1$, and $g(x)$ satisfy the following congruence equation 	
		\begin{equation}\label{key121}
			f(x)g(x) \equiv h(x) \pmod {x^k}.
		\end{equation} 
		Then if $\ell+\ell_1\leq k$, there exists a unique $g(x)$ with the degree no more than $\ell-1$ satisfying (\ref{key121}).
		%as
		%\begin{equation}
		%g(x) = \frac{h(x)}{f(x)} \pmod {x^\ell}.
		%\end{equation}
	\end{theorem}
	
	\begin{proof}
		As $\ell_1$ is the maximal integer such that $x^{\ell_1}\mid f(x)$, then $f(x)$ can be written $f(x)=x^{\ell_1}f_1(x)$, where $f_1(x)=\frac{f(x)}{x^{\ell_1}}\in F_p[x]$ and the constant term of $f_1(x)$ is nonzero. Taking $x^{\ell_1}f_1(x)$ to replace $f(x)$ in (\ref{key121}), then
		\begin{equation}\label{key123}
			x^{\ell_1}f_1(x)g(x) \equiv h(x) \pmod {x^k}.
		\end{equation}
		Thus, we infer that  $x^{\ell_1}$ is a factor of $h(x)$, then 
		\begin{equation}\label{key124}
			x^{\ell_1}f_1(x)g(x) \equiv x^{{\ell}_1}h_1(x)\pmod {x^{\ell_1}x^{k-{\ell}_1}},
		\end{equation} 
		where $h_1(x)=\frac{h(x)}{x^{\ell_1}}\in F_p[x]$. Combining the conclusion of Lemma~\ref{lem2}, we have
		\begin{equation}\label{key125}
			f_1(x)g(x) \equiv h_1(x) \pmod {x^{k-\ell_1}}.
		\end{equation} 
		As $\ell \leq k-\ell_1$, (\ref{key125}) can be simplified into
		\begin{equation}\label{key126}
			f_1(x)g(x) \equiv h_1(x) \pmod {x^{\ell}}.
		\end{equation}
		Since the constant term of $f_1(x)$ is nonzero, by using Lemma~\ref{lem3}, there exists the inverse of $f_1(x)$ in $F_p[x]/x^{\ell}$, then we have
		\begin{equation}\label{key127}
			g(x) \equiv h_1(x){f_1(x)}^{-1}=\frac{h(x)/x^{\ell_1}}{f(x)/x^{\ell_1}} \pmod {x^{\ell}},
		\end{equation}
		therefore, considering in the polynomial ring $F_p[x]/x^{\ell}$, there exists a unique solution for $g(x)$ satisfying (\ref{key121}). %For the convenience, we denote by $g(x)=\frac{h(x)}{f(x)}\pmod {x^{\ell}}$ the solution of (\ref{key121}) throughout this paper.
	\end{proof}
	
	\noindent\textbf{The proof of Correctness.} For $t\in \mathbb{N}^+, A \in\mathcal{A}_t$, then $|A|\geq 2$. We need to show that the $\ell$-bit secret $s$ can be correctly reconstructed by the shares of the participants in $A$. 
	Since the cases of $|A|\geq 2$ include the case of $|A|=2$, we only prove the case of $|A|=2$. 
	
	Without loss of generality, we take two participants $P_i$ and $P_j$ from $A$ with $i<j\leq t$. Since $y_i$ and $y_j$ are the binary prefix codes of $i$ and $j$, $l_i$ and $l_j$ are the code length of $i$ and $j$, thus we have $l_i\leq l_j$. Then, the $i$-th and $j$-th shares are as follows.
	\begin{equation}\label{key128}
		\begin{aligned}
			\left
			\{
			\begin{array}{c}
				Z^{(s)}_i=r_0+sy_i \pmod {x^{\ell_i+\ell-1}},\\
				Z^{(s)}_j=r_0+sy_j \pmod {x^{\ell_j+\ell-1}}.\\
			\end{array} \right.
		\end{aligned}
	\end{equation}
	By Lemma~\ref{lem4}, subtracting the second equation from the first equation in (\ref{key128}), we can further get
	\begin{equation}\label{key129}
		(y_i-y_j)s=Z^{(s)}_i-Z^{(s)}_j \pmod { x^{\ell_i+\ell-1}}.
	\end{equation}
	As $c_i$ is not a prefix of $c_j$, we have 
	\begin{equation}\label{key30}
		y_i-y_j\neq0 \pmod { x^{\ell_i}}.
	\end{equation}
	As $L_{i,j}$ denotes the maximal integer satisfying 
	$x^{L_{i,j}}\mid (y_i-y_j)$, combining the result of (\ref{key30}), we infer  
	\begin{equation}\label{key31}
		L_{i,j}\leq \ell_i-1.
	\end{equation}
	Thus,
	\begin{equation}
		L_{i,j}+\ell\leq \ell_i+\ell-1.
	\end{equation}
	Using the conclusion of Theorem~\ref{thm3} and combining the bit length of $s$ is $\ell$, the congruence equation (\ref{key129}) has a unique solution in $F_2[x]/x^\ell$, which can be calculated as 
	\begin{equation}
		s=\frac{(Z^{(s)}_i-Z^{(s)}_j)/x^{L_{i,j}}}{(y_i-y_j)/x^{L_{i,j}}}\pmod {x^\ell}.
	\end{equation}
	
	\noindent\textbf{The proof of Secrecy.} For any $t\in \mathbb{N}^+, C\in 2^{\mathcal{P}_t}\setminus\mathcal{A}_t$, i.e. $|C|<2$. We need to prove that the secret $s$ is unable to be recovered by the shares in $C$. Since the case of
	$|C|=0$ is trivial, we only prove the case of $|C|=1$ in the following. 
	
	Assume the only element in $C$ as $P_i$ with $i\leq t$.
	For any $s$, we have $Z^{(s)}_i=r_0+sy_i$ in $F_2[X]/(x^{\ell_i+\ell-1})$. As $r_0$ is a random variable uniformly distributed in the additive group $F_2[X]/(x^{\ell_i+\ell-1})$, $Z^{(s)}_i$ is independent from $sy_i$, which makes $Z^{(s)}_i$ independent from $s$. Hence, $Z^{(s)}_i$ is uniformly random in $F_2[X]/(x^{\ell_i+\ell-1})$ for each selection of $s$.
	
	Now, we will provide an example to show the processes of distributing shares and reconstructing secret.
	
	\noindent\textbf{Example.} Given a secret $s=1001$, for the two participants $P_3$, $P_4$, let $101$ and $11000$ be the binary prefix codes of $3$ and $4$, respectively. According to $\ell+\ell_4-1=8$, the algorithm $\mathcal{E}$ randomly chooses a $8$-bit binary string $r_0=10101001$. Then the share of $P_3$ is given by
	\begin{align*}
		Z^{(s)}_3=&r_0+sy_3=(1+x^2+x^4)+(1+x^3)(1+x^2)\nonumber\\
		=&x^3+x^4+x^5 \pmod {x^6},
	\end{align*}
	thus, the $3$-th share $Z^{(s)}_3$ is $000111$.
	
	And the share of $P_4$ is given by
	\begin{align*}
		Z^{(s)}_4=&r_0+sy_4=(1+x^2+x^4+x^7)+(1+x^3)(1+x)\nonumber\\
		=&x+x^2+x^3+x^7 \pmod {x^8},
	\end{align*}
	thus, the $4$-th share $Z^{(s)}_4$ is $01110001$.
	
	Next, we take this example to show how to reconstruct the secret $s$. Let the bit length of $s$ be $4$. Let $101$ and $11000$ be the prefix codes of $3$ and $4$, respectively. The shares of the two participants $P_3$ and $P_4$ are $000111$ and $01110001$, respectively. Then we have 
	\begin{equation*}
		\begin{aligned}
			\left
			\{
			\begin{array}{l}
				Z^{(s)}_3=x^3+x^4+x^5 =r_0+s(1+x^2) \pmod {x^6},\\
				Z^{(s)}_4=x+x^2+x^3+x^7=r_0+s(1+x) \pmod {x^8}.\\
			\end{array} \right.
		\end{aligned}
	\end{equation*}
	Since $\ell+\ell_3-1=6$, $y_3-y_4=x+x^2=x(1+x)$ and the bit length of $s$ is 4, the algorithm $\mathcal{R}$ solves the following equation 
	\begin{align*}
		s(y_3-y_4)=Z^{(s)}_3-Z^{(s)}_4 \pmod {x^6}
	\end{align*}
	with a unique solution, i.e.
	\begin{align*}
		s=&\frac{(Z^{(s)}_3-Z^{(s)}_4)/x}{(y_3-y_4)/x}=\frac{(x+x^2+x^4+x^5+x^7)/x}{(x+x^2)/x} \pmod {x^4}\nonumber\\
		=&\frac{1+x+x^3}{1+x} \pmod {x^4}\nonumber\\
		=&(1+x+x^3)(1+x)^{-1} \pmod {x^4} \nonumber\\
		\overset{(a)}=&(1+x+x^3)(1+x+x^2+x^3)\pmod {x^4} \nonumber\\
		=&1+x^3.
	\end{align*}
	where (a) holds since $(1+x)(1+x+x^2+x^3)=1$ in $F_2[x]/x^4$. Therefore, the algorithm $\mathcal{R}$ outputs $s=1001$, which is correct.
	
	%\IEEEpubidadjcol 换行
	
	\subsection{The Share Size}\label{sub3.3}
	We analyze the share size in the proposed scheme. For the $t$-th participant, the share $Z^{(s)}_t$ can be regarded as a polynomial of $x$ in $F_2$ with the degree no more than $\ell_t+\ell-2$ from (\ref{key105}), where $\ell_t$ is the length of binary prefix code for the positive integer $t$. Then the size of the corresponding share satisfies
	\begin{equation}\label{key1}
		B(Z^{(s)}_t)=\ell_t+\ell-1.	
	\end{equation}
	From (\ref{key1}), we find that it may obtain different share sizes when choosing different binary prefix codes.
	
	Next, we show the corresponding share size by introducing several binary prefix codes. We denote by $L(\cdot)$ the codeword length for encoding $t$ for $t\in\mathbb{N}^+$. Consider $\gamma$ code~\cite{elias1975universal}, a widely used integer universal coding, which is represented by
	\begin{equation*}
		\gamma(t)=B_{U}(\lfloor \lg{t}\rfloor)[t]_2,
	\end{equation*}
	where $B_{U}$ is written as
	\begin{equation*}
		B_{U}(t)=\overbrace{00\cdots0}^{t}
	\end{equation*}
	and $[t]_2$ denotes the binary expression. Then the codeword length of $\gamma(t)$ is given by
	\begin{equation*}
		L(\gamma(t))=2\lfloor\lg t\rfloor+1.
	\end{equation*}
	Therefore, when using $\gamma$ code~\cite{elias1975universal} as the binary prefix code, the size of the $t$-th share satisfies
	\begin{equation}\label{key38}
		B(Z^{(s)}_t)=2\lfloor\lg t\rfloor+\ell.		
	\end{equation}
	
	Except for $\gamma$ code, we also consider another binary prefix coding, $\delta$ coding~\cite{elias1975universal}, which is represented by
	\begin{equation*}
		\delta(t)={\gamma}'(\lfloor \lg{t}\rfloor+1)[t]'_2,
	\end{equation*}
	where ${\gamma}'$ is the variant of $\gamma$ code. It is defined as
	\begin{equation*}
		{\gamma}'(t)=\overbrace{00\cdots0}^{\lfloor \lg{t}\rfloor}1[t]'_2
	\end{equation*}
	where $[t]'_2$ denotes the binary string deleting the most significant bit of $[t]_2$. Then the codeword length of $\delta(t)$ is given by
	\begin{equation*}
		L(\delta(t))=\lfloor\lg t\rfloor+2\lfloor\lg ({\lfloor\lg t\rfloor+1})\rfloor+1 .
	\end{equation*}
	Therefore, if using $\delta$ code as the binary prefix code, the share size of the $t$-th satisfies
	\begin{equation}\label{key39}
		B(Z^{(s)}_t)=\lfloor\lg t\rfloor+2\lfloor\lg ({\lfloor\lg t\rfloor+1})\rfloor+\ell.		
	\end{equation}
	Compared with the scheme~\cite{komargodski2017share}, when $\ell=1$, the result of (\ref{key39}) is approximately equal to the result given by Theorem 1. When $\ell\geq 2$, the result is smaller than the result given by Theorem 1.
	
	\noindent\textbf{Discussion.} We will discuss which encoding method can achieve a lower share size for the two proposed binary prefix codes. Let $f(t)=(38)-(39)$, then we have 
	\begin{equation}\label{key245}
		f(t)=\lfloor\lg t\rfloor-2\lfloor\lg ({\lfloor\lg t\rfloor+1})\rfloor.		
	\end{equation}
	Hence, the problem becomes to compare the relationship between $f(t)$ and $0$. We classify the problem into two cases to discuss, (\romannumeral1) $t=1$; (\romannumeral2) $t\geq 2$.
	
	\textbf{Case 1.} When $t=1$, we can directly calculate $f(1)=0$.
	
	\textbf{Case 2.} When $t\geq2$, we first introduce a unique representation method for $t$. Let $t=2^{2^x+y}+z$, where $x\geq 0$, $0\leq y\leq 2^x-1$ and $0\leq z\leq 2^{2^x+y}-1$. Substituting $t$ by $2^{2^x+y}+z$ in (\ref{key245}), we can simplify (\ref{key245}) further as below
	\begin{equation}\label{key246}
		f(t)=f(x,y,z)=2^x+y-2\lfloor\lg (2^x+y+1)\rfloor,	
	\end{equation}
	furthermore, the value of $f(x,y,z)$ is analyzed as follows. For any $z$ with $0\leq z\leq 2^{2^x+y}-1$, we have 
	\begin{equation}\label{key247}
		f(x,y,z)=
		\begin{aligned}				
			\left\{
			\begin{array}{ll}
				2^x+y-2x &\enspace \textrm{if } 0 \leq y<2^x-1,\\
				2^{x+1}-2x-3 &\enspace\textrm{if } y=2^x-1.\\			
			\end{array} \right.
		\end{aligned}
	\end{equation}
	By analyzing the value of $y$, we further get
	\begin{equation}\label{key248}
		\begin{aligned}				
			\left\{
			\begin{array}{ll}
				f(x,y,z)>0 &\enspace\textrm{if } y=0, x\geq3,\\
				f(x,y,z)=0 & \enspace\textrm{if } y=0, 0< x<3,\\
				f(x,y,z)>0 &\enspace \textrm{if } 0 < y<2^x-1, x\geq 0,\\
				f(x,y,z)>0 & \enspace \textrm{if }  y=2^x-1, x\geq 2,\\
				f(x,y,z)<0 & \enspace\textrm{if }  y=2^x-1, 0\leq x<2.\\			
			\end{array} \right.
		\end{aligned}
	\end{equation}
	Combining the results of Case 1 and Case 2,  we summarize the above results below
	\begin{equation}\label{key249}
		\begin{aligned}				
			\left\{
			\begin{array}{ll}
				f(t)=0 &\enspace\textrm{if } t=1 \textrm{ or } 4\leq t\leq 7\textrm{ or }16\leq t\leq 31,\\
				f(t)<0 &\enspace\textrm{if } 2\leq t\leq 3\textrm{ or }8\leq t\leq 15,\\
				f(t)>0 &\enspace\textrm{if } t\geq 32.\\			
			\end{array} \right.
		\end{aligned}
	\end{equation}
	Therefore, using $\delta$ code can achieve a lower share size than $\gamma$ code when $t\geq 32$. 
	
	\section{Evolving $k$-threshold Secret Sharing Scheme over $F_2[x]$}\label{sec4}
	Based on similar techniques, we extend the prior scheme to evolving $k$-threshold scheme in this section. We will give the construction of the evolving $k$-threshold scheme for an $\ell$-bit secret $s$ on $F_2[x]$, where $k\geq 3$. The scheme allows the dealer to distribute the share to each participant such that only no less than $k$ participants can reconstruct $s$.
	
	\subsection{Proposed Scheme}\label{sec4.1}
	For any $i\in \mathbb{N}^+$, then the $i$-th share with the algorithm $\mathcal{E}$ is defined as
	\begin{equation}\label{key243}
		Z^{(s)}_i=\sum_{j=0}^{k-2}r_j{y_i}^j+s{y_i}^{k-1}\pmod {x^{(\ell_i-1)(k-1)+\ell}},
	\end{equation}
	where $r_j$ is randomly chosen from $F_2[[x]]$ for $0\leq j\leq k-2$. However, we cannot choose a $r_j\in F_2[[x]]$ in practice. Due to the operation modulo $x^{(\ell_i-1)(k-1)+\ell}$ in \eqref{key243}, the dealer only needs to choose a part of $r_j$ with degrees less than $(\ell_i-1)(k-1)+\ell$. 
	
	Given any $k$ participants $P_{i_1}, P_{i_2},\cdots, P_{i_k}\in \mathcal{P}$ with $i_1<\cdots<i_{k}$, let $L_{i_u,i_v}$ denote the maximal integer of $y_{i_u}-y_{i_v}$ satisfying $x^{L_{i_u,i_v}}\mid (y_{i_u}-y_{i_v})$, for any $u,v\in[k]$ with $u\neq v$. Let $\alpha=\min_{m=1}^{k}\{L_{i_m}+\sum_{\substack{1\leq u<v \leq k\\u,v\neq m}} L_{i_u,i_v}\}$,  where $L_{i_m}=(\ell_{i_m}-1)(k-1)+\ell$ for $1\leq m\leq k$.
	The algorithm $\mathcal{R}$ finds out that the following equation  
	\begin{equation}\label{key244}
		s\Big(\prod_{1\leq u<v \leq k} (y_{i_u}-y_{i_v})\Big)=\sum_{m=1}^{k}\prod_{\substack{1\leq u<v \leq k\\u,v\neq m}} (y_{i_u}-y_{i_v})Z^{(s)}_{i_m}\pmod {x^\alpha},
	\end{equation}
	has a unique solution of $s$ in $F_2[x]/x^\ell$ as
	\begin{align}\label{key2441}
		s=\frac{\sum_{m=1}^{k}\prod_{\substack{1\leq u<v \leq k\\u,v\neq m}} (y_{i_u}-y_{i_v})Z^{(s)}_{i_m}/x^{\sum_{1\leq u<v \leq k}L_{i_u,i_v}}}{\prod_{1\leq u<v \leq k} (y_{i_u}-y_{i_v})/x^{\sum_{1\leq u<v \leq k}L_{i_u,i_v}}}\pmod {x^\ell}.	
	\end{align}
	
	Similarly, considering in $F_2[X]/x^{\ell}$, the existence and uniqueness of $s$ will be provided in the proof of correctness in next subsection.
	
	When the dealer generates the $i$-th share via \eqref{key243}, where $i$ is sufficient large, the dealer also needs to construct the corresponding random polynomials $\{r_j\}_{j=0}^{k-2}$. The dealer still first distributes secret $s$ to the participants in the group. When a new participant joins the group, the dealer then generates and assigns new shares for $s$ to the new participant. Without loss of generality, we suppose that there exists $k$ participants in the group initially. The dealer first randomly chooses the suitable random polynomials $\{r_j\}_{j=0}^{k-2}$, where the degree of each $r_j$ is no less than $(\ell_k-1)(k-1)+\ell-1$. Based on these selected $\{r_j\}_{j=0}^{k-2}$, the dealer distributes the secret $s$ to $k$ participants via \eqref{key243}. After completing this distribution, the dealer discards the selected $\{r_j\}_{j=0}^{k-2}$ and $s$. When the new participant joins the group, which is denoted as $P_{k+1}$, to distribute share to $P_{k+1}$ by \eqref{key243}, the corresponding $\{r_j\}_{j=0}^{k-2}$ and $s$ in \eqref{key243} are necessary to know. According to the previous known shares $\{Z^{(s)}_i\}_{i=1}^{k}$ and $\{y_i\}_{i=1}^{k}$, the dealer can reconstruct the secret $s$ by calculating (\ref{key2441}), then calculates the coefficients of $\{r_j\}_{j=0}^{k-2}$ with degree less than $(\ell_k-1)(k-1)+\ell$ by solving the following equations 
	\begin{equation}\label{key3148}
		\begin{aligned}
			\left
			\{
			\begin{array}{c}
				\sum_{j=0}^{k-2}r_j{y_{1}}^j=	Z^{(s)}_1-s{y_{1}}^{k-1}\pmod {x^{(\ell_{1}-1)(k-1)+\ell}},\\
				\sum_{j=0}^{k-2}r_j{y_2}^j=	Z^{(s)}_2-s{y_2}^{k-1}\pmod {x^{(\ell_2-1)(k-1)+\ell}},\\
				\cdots\\
				\sum_{j=0}^{k-2}r_j{y_k}^j=	Z^{(s)}_k-s{y_k}^{k-1}\pmod {x^{(\ell_k-1)(k-1)+\ell}}.\\
			\end{array} \right.
		\end{aligned}
	\end{equation}
	Solving the above equations, there must exist the solution $\{r_j\}_{j=0}^{k-2}$ with the degree less than $(\ell_{k}-1)(k-1)+\ell$ since the initial $\{r_j\}_{j=0}^{k-2}$ that the dealer chooses for the first time can make \eqref{key3148} hold. Then the dealer randomly chooses $(\ell_{k+1}-\ell_{k})(k-1)$ coefficients such that each $r_j$'s degree is $(\ell_{k+1}-1)(k-1)+\ell-1$, then the corresponding $\{r_j\}_{j=0}^{k-2}$ for calculating $Z^{(s)}_{k+1}$ is obtained. For the $m$-th participant for $m>k$, the corresponding $\{r_j\}_{j=0}^{k-2}$ can aslo be obtained by solving the similar equations established by the shares $\{Z^{(s)}_i\}_{i=1}^{m-1}$ and $\{y_i\}_{i=1}^{m-1}$. To avoid writing repetition, the algorithm to reconstruct $\{r_j\}_{j=0}^{k-2}$ won't be described in this paper.
	
	Next, we will discuss the corresponding share size under different binary prefix codes. The share size in the proposed evolving $k$-threshold scheme is analyzed as follows.
	\\
	\noindent\textbf{Share Size.} For the $t$-th participant, the share $Z^{(s)}_t$ is a polynomial of $x$ in $F_2$ and the degree is  $(k-1)(\ell_t-1)+\ell-1$ from (\ref{key243}), where $\ell_t$ is the length of prefix code to encode positive integer $t$. Hence, the share size satisfies
	\begin{equation}
		B(Z^{(s)}_{t})=(k-1)(\ell_t-1)+\ell.	
	\end{equation}
	We have described two binary prefix codes in Subsection~\ref{sub3.3}. If using $\gamma$ code as the prefix code, the corresponding share size satisfies
	\begin{equation}
		B(Z^{(s)}_t)=2(k-1)\lfloor\lg t\rfloor+\ell.		
	\end{equation}
	However, when using $\delta$ code as the prefix code, then the corresponding share size satisfies
	\begin{equation}\label{key3153}
		B(Z^{(s)}_t)=(k-1)\lfloor\lg t\rfloor+2(k-1)\lfloor\lg ({\lfloor\lg t\rfloor+1}) \rfloor+\ell.		
	\end{equation}
	Compared with the scheme~\cite{komargodski2017share}, when $k\geq 3$, the above result of (\ref{key3153}) improves the result of Theorem 2.
	
	\subsection{Comparision}
	In this subsection, we tabulate the share sizes of currently known evolving threshold secret sharing schemes. For the proposed scheme, we show the corresponding share size for using $\delta$ code as the prefix code. In Table~\ref{tab2}, we represent the value of the lowest share size in bold for each case of $k$. When $k=2$, $k=3$ and $k\geq 4$,  the proposed schemes can achieve consistently lower share sizes than the scheme \cite{komargodski2017share}. When $k=3$, the proposed scheme' share size is larger than the result of \cite{DARCO2021149}. However, the result of \cite{DARCO2021149} is optimized only for the case $k=3$.
	\begin{table}[t]
		\renewcommand\arraystretch{1}
		\centering
		\caption{Share sizes of evolving $k$-threshold schemes.}
		\setlength{\tabcolsep}{1.5mm}
		\begin{tabular}{|c|c|c|}
			\hline
			threshold &algorithm &share size\\\hline
			\multirow{2}*{$k=2$}
			&~\cite{komargodski2017share}&$\lg t+(\ell+1)\lg{\lg t}+4\ell+1$\\
			\cline{2-3}
			~& \textbf{ours}& \bm {${\lfloor\lg t\rfloor+2\lfloor\lg ({\lfloor\lg t\rfloor+1})\rfloor+\ell}$}\\
			\hline
			\multirow{3}*{$k=3$}
			&\cite{komargodski2017share}&$2\lg t+486\ell\lg{\lg t}\cdot\lg{\lg {\lg t}}+ 567\ell\lg 3$\\
			\cline{2-3}
			~& \textbf{\cite{DARCO2021149}}& \bm {${\frac{4}{3}\lg t+c(\log_4{\lg t})^2+\lg p(\log_4{\lg t})}$} \\
			\cline{2-3}
			~&ours& $2 \lfloor\lg t\rfloor+4\lfloor\lg ({\lfloor\lg t\rfloor+1}) \rfloor+\ell$\\
			\hline
			\multirow{2}*{$k\geq 4$}
			&~\cite{komargodski2017share}&$(k-1)\lg t+6k^4\ell\lg{\lg t}\cdot\lg{\lg {\lg t}}+ 7k^4\ell\lg k$\\
			\cline{2-3}
			~&\textbf{ours}&\bm {${(k-1)\lfloor\lg t\rfloor+2(k-1)\lfloor\lg ({\lfloor\lg t\rfloor+1}) \rfloor+\ell}$}\\
			\hline				
		\end{tabular}
		\label{tab2}
	\end{table} 
	
	\subsection{Proofs of Correctness and Secrecy}\label{sec4.2}
	The proofs of the correctness and secrecy of the proposed scheme are given as follows.
	
	\noindent\textbf{The proof of Correctness.} For any $t\in \mathbb{N}^+, A \in\mathcal{A}_t$, then $|A|\geq k$, we will prove that the secret $s$ can be correctly recovered by the shares of the participants in $A$. Since the cases of $|A|\geq k$ include the case of $|A|=k$, we only prove the case of $|A|=k$. 
	
	Denote the $k$ elements in $A$ as $P_{i_1}$, $P_{i_2}$, $\cdots$, $P_{i_k}$ with $i_1<\cdots<i_k\leq t$. Since $y_{i_m}$ is the prefix code of $i_m$ and $l_{i_m}$ is the code length of $i_m$, thus $i_m$ is increasing about $m$, where $1\leq m\leq k$. Then the corresponding shares are as follows.
	\begin{equation}\label{key348}
		\begin{aligned}
			\left
			\{
			\begin{array}{c}
				Z^{(s)}_{i_1}=\sum_{j=0}^{k-2}r_j{y_{i_1}}^j+s{y_{i_1}}^{k-1}\pmod {x^{(\ell_{i_1}-1)(k-1)+\ell}},\\
				Z^{(s)}_{i_2}=\sum_{j=0}^{k-2}r_j{y_{i_2}}^j+s{y_{i_2}}^{k-1}\pmod {x^{(\ell_{i_2}-1)(k-1)+\ell}},\\
				\cdots\\
				Z^{(s)}_{i_k}=\sum_{j=0}^{k-2}r_j{y_{i_k}}^j+s{y_{i_k}}^{k-1}\pmod {x^{(\ell_{i_k}-1)(k-1)+\ell}}.\\
			\end{array} \right.
		\end{aligned}
	\end{equation} 
	In order to calculate $s$, we hope to eliminate these elements $r_0,\cdots, r_{k-2}$. Considering each congruence equation in (\ref{key348}), there exists $h_m(x)\in F_2[x]$ such that the equation 
	\begin{equation}\label{key349}
		\sum_{j=0}^{k-2}r_j{y_{i_m}}^j+s{y_{i_m}}^{k-1}-Z^{(s)}_{i_m}=h_m(x)x^{(\ell_{i_m}-1)(k-1)+\ell}	
	\end{equation}
	holds. 
	
	Considering the above equation (\ref{key349}),
	we multiply both sides of the equation by ${\prod_{\substack{1\leq u<v \leq k\\u,v\neq m}}}(y_{i_u}-y_{i_v})$. For the convenience of writing, let $H_{i_m}(x)={\prod_{\substack{1\leq u<v \leq k\\u,v\neq m}}}(y_{i_u}-y_{i_v})$, then the equation (\ref{key349}) becomes 
	\begin{equation}
		H_{i_m}(x)\Big(\sum_{j=0}^{k-2}r_j{y_{i_m}}^j+s{y_{i_m}}^{k-1}-Z^{(s)}_{i_m}\Big)=H_{i_m}(x)h_m(x)x^{(\ell_{i_m}-1)(k-1)+\ell}.	
	\end{equation}
	Performing the above same steps for each congruent equation in (\ref{key348}), then we can obtain
	\begin{equation}\label{key351}
		\begin{aligned}
			\left
			\{
			\begin{array}{c}
				H_{i_1}(x)\Big(\sum_{j=0}^{k-2}r_j{y_{i_1}}^j+s{y_{i_1}}^{k-1}-Z^{(s)}_{i_1}\Big)=H_{i_1}(x)h_1(x){x^{(\ell_{i_1}-1)(k-1)+\ell}},\\
				H_{i_2}(x)\Big(\sum_{j=0}^{k-2}r_j{y_{i_2}}^j+s{y_{i_2}}^{k-1}-Z^{(s)}_{i_2}\Big)=H_{i_2}(x)h_2(x)x^{(\ell_{i_2}-1)(k-1)+\ell},\\
				\cdots\\
				H_{i_k}(x)\Big(\sum_{j=0}^{k-2}r_j{y_{i_k}}^j+s{y_{i_k}}^{k-1}-Z^{(s)}_{i_k}\Big)=H_{i_k}(x)h_k(x)x^{(\ell_{i_k}-1)(k-1)+\ell},\\
			\end{array} \right.
		\end{aligned}	
	\end{equation}
	where $H_{i_m}(x)={\prod_{\substack{1\leq u<v \leq k\\u,v\neq m}}}(y_{i_u}-y_{i_v})$.
	
	Summing all equations in \eqref{key351}, we have
	\begin{align}\label{key3520}
		&\sum_{m=1}^{k}H_{i_m}(x){y_{i_m}}^{k-1}s+\sum_{j=0}^{k-2} \sum_{m=1}^{k}H_{i_m}(x){y_{i_m}}^{j}r_j\nonumber\\
		=& \sum_{m=1}^{k}H_{i_m}(x)Z^{(s)}_{i_m}+\sum_{m=1}^{k}H_{i_m}(x)h_m(x)x^{(\ell_{i_m}-1)(k-1)+\ell}.
	\end{align}
	Consider the coefficient of $s$ as $\sum_{m=1}^{k}H_{i_m}(x){y_{i_m}}^{k-1}$ in $F_2[x]$, which is equal to the following Vandermonde determinant, i.e.
	\begin{equation}\label{key3521}
		\sum_{m=1}^{k}H_{i_m}(x){y_{i_m}}^{k-1}
		\overset{(b)}=\left |\begin{array}{ccccc}
			1 &1   &\cdots &1\\
			y_{i_1} &y_{i_2} &\cdots &y_{i_{k}}\\
			\vdots &\vdots &\ddots &\vdots\\
			{y_{i_1}}^{k-1} & {y_{i_2}}^{k-1} &\cdots &{y_{i_{k}}}^{k-1} \\
		\end{array}\right|=\prod_{1\leq u<v \leq k} (y_{i_u}-y_{i_v}),
	\end{equation}
	where (b) holds since $\sum_{m=1}^{k}H_{i_m}(x){y_{i_m}}^{k-1}$ is equal to the result of expanding the above determinant based on the last row. 
	
	For arbitrary $j$ with $0\leq j\leq k-2$, the coefficient of $r_j$ is $\sum_{m=1}^{k}H_{i_m}(x){y_{i_m}}^{j}$, then we have
	\begin{equation}\label{key3522}
		\sum_{m=1}^{k}H_{i_m}(x){y_{i_m}}^{j}
		\overset{(c)}=\left |\begin{array}{ccccc}
			1 &1   &\cdots &1\\
			y_{i_1} &y_{i_2} &\cdots &y_{i_{k}}\\
			\vdots &\vdots &\ddots &\vdots\\
			{y_{i_1}}^{k-2} & {y_{i_2}}^{k-2} &\cdots &{y_{i_{k}}}^{k-2} \\
			{y_{i_1}}^{j} & {y_{i_2}}^{j} &\cdots &{y_{i_{k}}}^{j} \\
		\end{array}\right|=0,
	\end{equation}
	where (c) holds since $\sum_{m=1}^{k}H_{i_m}(x){y_{i_m}}^{j}$ is equal to the result of expanding the above determinant based on the last row. 
	
	Taking the results of \eqref{key3521} and \eqref{key3522} into \eqref{key3520},  we can further get
	\begin{equation}\label{key352}
		\prod_{1\leq u<v \leq k} (y_{i_u}-y_{i_v})s=\sum_{m=1}^{k}H_{i_m}(x)Z^{(s)}_{i_m}+\sum_{m=1}^{k}H_{i_m}(x)h_m(x)x^{(\ell_{i_m}-1)(k-1)+\ell}.
	\end{equation} 
	Since $L_{i_u,i_v}$ denotes the maximal integer of $y_{i_u}-y_{i_v}$ satisfying $x^{L_{i_u,i_v}}\mid (y_{i_u}-y_{i_v})$ for any $u,v\in[k]$ with  $u\neq v$, and $L_{i_m}=(\ell_{i_m}-1)(k-1)+\ell$ for $1\leq m\leq k$, we can infer that each polynomial
	\begin{equation*}
		H_{i_m}(x)h_m(x)x^{L_{i_m}}={\prod_{\substack{1\leq u<v \leq k\\u,v\neq m}}}(y_{i_u}-y_{i_v})x^{L_{i_m}}h_m(x)
	\end{equation*}
	has the factor $x^\alpha$ for any $m$, where $\alpha=\min_{m=1}^{k}\{L_{i_m}+\sum_{\substack{1\leq u<v \leq k\\u,v\neq m}} L_{i_u,i_v}\}$.
	
	Substituting $H_{i_m}(x)$ by ${\prod_{\substack{1\leq u<v \leq k\\u,v\neq m}}(y_{i_u}-y_{i_v})}$ in \eqref{key352} and simplying the equation further, then we get 
	\begin{equation}\label{key353}
		s\Big(\prod_{1\leq u<v \leq k} (y_{i_u}-y_{i_v})\Big)\equiv \sum_{m=1}^{k}\prod_{\substack{1\leq u<v \leq k\\u,v\neq m}} (y_{i_u}-y_{i_v})Z^{(s)}_{i_m} \pmod {x^{\alpha}}.
	\end{equation}
	We note that the polynomial $\prod_{1\leq u<v \leq k} (y_{i_u}-y_{i_v})$ has the maximal integer $\sum_{1\leq u<v \leq k}L_{i_u,i_v}$, i.e. 
	\begin{equation}
		x^{\sum_{1\leq u<v \leq k}L_{i_u,i_v}}\mid \prod_{1\leq u<v \leq k} (y_{i_u}-y_{i_v}).
	\end{equation}
	If $\sum_{1\leq u<v \leq k}L_{i_u,i_v}+\ell\leq \alpha$, we can use the conclusion of Theorem~\ref{thm3} to construct $s$. Next, our goal is to prove $\sum_{1\leq u<v \leq k}L_{i_u,i_v}+\ell\leq \alpha$.
	
	Without loss of generality, we suppose $\alpha=\min_{m=1}^{k}\{L_{i_m}+\sum_{\substack{1\leq u<v \leq k\\u,v\neq m}}  L_{i_u,i_v}\}=L_{i_q}+\sum_{\substack{1\leq u<v \leq k\\u,v\neq q}} L_{i_u,i_v}$ for some $q$. Since $L_{i_q}=(\ell_{i_q}-1)(k-1)+\ell$, then we have
	\begin{align}\label{key354}
		&\alpha-\sum_{1\leq u<v \leq k}L_{i_u,i_v}-\ell\nonumber\\
		=&L_{i_q}+\sum_{\substack{1\leq u<v \leq k\\u,v\neq q}} L_{i_u,i_v}-\sum_{1\leq u<v \leq k}L_{i_u,i_v}-\ell\nonumber\\
		=&(\ell_{i_q}-1)(k-1)-\sum_{\substack{1\leq u\leq k\\u\neq q}} L_{i_u,i_q}\nonumber\\
		=&\sum_{\substack{1\leq u\leq k\\u\neq q}}( \ell_{i_q}-1-L_{i_u,i_q}).
	\end{align}
	For any $u\neq q$ with $1\leq u\leq k$, since $y_{i_u}$ is the prefix code of $i_u$, then $y_{i_u}-y_{i_q}\neq 0$. On the other hand, as $L_{i_u,i_q}$ is the maximal integer of $y_{i_u}-y_{i_q}$ with $x^{L_{i_u,i_q}}\mid y_{i_u}-y_{i_q}$, and $\ell_{i_q}$ is the prefix codeword length of $y_{i_q}$, then we have $L_{i_u,i_q}\leq \ell_{i_q}-1$. Replacing the result in (\ref{key354}), we get 
	\begin{equation}
		\sum_{1\leq u<v \leq k}L_{i_u,i_v}+\ell\leq \alpha.	
	\end{equation}
	By using the conclusion of Theorem~\ref{thm3}, we can construct the unique solution of $s$ of (\ref{key353}) in $F_2[x]/x^\ell$ as
	\begin{align}\label{key355}
		s=\frac{\sum_{m=1}^{k}\prod_{\substack{1\leq u<v \leq k\\u,v\neq m}} (y_{i_u}-y_{i_v})Z^{(s)}_{i_m}/x^{\sum_{1\leq u<v \leq k}L_{i_u,i_v}}}{\prod_{1\leq u<v \leq k} (y_{i_u}-y_{i_v})/x^{\sum_{1\leq u<v \leq k}L_{i_u,i_v}}}\pmod {x^\ell}.	
	\end{align}
	Therefore, the correctness of the proposed scheme has been proved completely.
	
	In the second part of this subsection, we will demonstrate the security of this scheme. Before proving the security, we first emphasize a theorem, which provides a useful conclusion for proving the security.
	\begin{theorem}\label{thm6}
		Let $i_1,i_2,\cdots,i_{k-1}\in \mathbb{N}^+$ satisfy $i_1<i_2<\cdots<i_{k-1}$, and $\ell_{i_1},\cdots,\ell_{i_{k-1}}\in \mathbb{N}^+$ satisfy $\ell_{i_1}\leq\ell_{i_2}\leq \cdots \leq \ell_{i_{k-1}}$. For any $m\in[k-1]$, let $y_{i_m}\in F_p[x]/x^{l_{i_m}}$ be the given polynomial, and $y_{i_{m_1}}$, $y_{i_{m_2}}$ be pairwisely different for $m_1\neq m_2$. For $\ell\in \mathbb{N}^+$, let $s_0, s_1\in F_p[x]/x^{\ell}$ be two different polynomials. Gvien the $k-1$ polynomials $z_{i_m}\in F_p[x]/x^{L_{i_m}}$ for $1\leq m \leq k-1$, where $L_{i_m}=(\ell_{i_m}-1)(k-1)+\ell$. Let $(r_{0,0}, r_{0,1},\cdots, r_{0,k-2})$ and $(r_{1,0}, r_{1,1},\cdots, r_{1,k-2})$ respectively be the solutions of the following two congruence equations
		\begin{equation}\label{key356}
			\begin{aligned}
				\left
				\{
				\begin{array}{l}
					z_{i_1}=\sum_{j=0}^{k-2}r_{0,j}{y_{i_1}}^j+s_0{y_{i_1}}^{k-1}\pmod {x^{L_{i_1}}},\\
					z_{i_2}=\sum_{j=0}^{k-2}r_{0,j}{y_{i_2}}^j+s_0{y_{i_2}}^{k-1}\pmod {x^{L_{i_2}}},\\
					\hspace{10em}\cdots\hspace{10em}\\
					z_{i_{k-1}}=\sum_{j=0}^{k-2}r_{0,j}{y_{i_{k-1}}}^j+s_0{y_{i_{k-1}}}^{k-1}\pmod {x^{L_{i_{k-1}}}},\\
				\end{array} \right.
			\end{aligned}
		\end{equation}
		and 
		\begin{equation}\label{key357}
			\begin{aligned}
				\left
				\{
				\begin{array}{l}
					z_{i_1}=\sum_{j=0}^{k-2}r_{1,j}{y_{i_1}}^j+s_1{y_{i_1}}^{k-1}\pmod {x^{L_{i_1}}},\\
					z_{i_2}=\sum_{j=0}^{k-2}r_{1,j}{y_{i_2}}^j+s_1{y_{i_2}}^{k-1}\pmod {x^{L_{i_2}}},\\
					\hspace{10em}\cdots\hspace{10em}\\
					z_{i_{k-1}}=\sum_{j=0}^{k-2}r_{1,j}{y_{i_{k-1}}}^j+s_1{y_{i_{k-1}}}^{k-1}\pmod {x^{L_{i_{k-1}}}},\\
				\end{array} \right.
			\end{aligned}
		\end{equation}
		where $r_{h,j} \in F_p[x]/x^{L_{i_{k-1}}}$ for $0\leq h\leq 1$, $0\leq j \leq k-2$.
		
		Then the equations \eqref{key356} and \eqref{key357} have the same number of solutions.	
	\end{theorem}
	\begin{proof}
		The proof is refered to the appendix.	
	\end{proof}
	
	\noindent\textbf{The proof of Secrecy.} For any $t\in \mathbb{N}^+, C \in 2^{\mathcal{P}_t}\setminus\mathcal{A}_t$, we will prove that the secret $s$ is unable to be recovered by the shares of participants in $C$. $C$ is unqualified, then $|C|<k$. Since the cases of $|C|<k$ include the case of $|C|=k-1$, we only prove the case of $|C|=k-1$. Other cases can be proved according to the case of $|C|=k-1$.
	
	Denote the elements in $C$ as $P_{i_1},\cdots, P_{i_{k-1}}$ with $i_1<\cdots<i_{k-1}\leq t$. In order to use the conclusion of Lemma~\ref{lem1} to prove the security of the proposed scheme, we choose any two distinct $s_0, s_1$, let $Z^{(s_0)}_{i_m}$ and $Z^{(s_1)}_{i_m}$ be corresponding the $i_m$-th shares for $1\leq m\leq k-1$. We need to prove the distributions of $\{Z^{(s_0)}_{i_m}=z_{i_m}\}_{m=1}^{k-1}$ and $\{Z^{(s_1)}_{i_m}=z_{i_m}\}_{m=1}^{k-1}$ are identical for any $z_{i_m}\in F_2[x]/x^{L_{i_m}}$, where $L_{i_m}=(\ell_{i_m}-1)(k-1)+\ell$ for $1\leq m\leq k-1$. It is equivalent to prove that the following two probabilities are equal, i.e.
	\begin{equation}\label{key358}
		P(\{Z^{(s_0)}_{i_m}=z_{i_m}\}_{m=1}^{k-1})=P(\{Z^{(s_1)}_{i_m}=z_{i_m}\}_{m=1}^{k-1}).
	\end{equation}
	
	We first analyze the value of $P(\{Z^{(s_0)}_{i_m}=z_{i_m}\}_{m=1}^{k-1})$. Consider the following congruence equations
	\begin{equation}\label{key359}
		\begin{aligned}
			\left
			\{
			\begin{array}{l}
				z_{i_1}=\sum_{j=0}^{k-2}r_{0,j}{y_{i_1}}^j+s_0{y_{i_1}}^{k-1}\pmod {x^{L_{i_1}}},\\
				z_{i_2}=\sum_{j=0}^{k-2}r_{0,j}{y_{i_2}}^j+s_0{y_{i_2}}^{k-1}\pmod{x^{L_{i_2}}},\\
				\hspace{10em}\cdots\hspace{10em}\\
				z_{i_{k-1}}=\sum_{j=0}^{k-2}r_{0,j}{y_{i_{k-1}}}^j+s_0{y_{i_{k-1}}}^{k-1}\pmod {x^{L_{i_{k-1}}}}.\\
			\end{array} \right.
		\end{aligned}
	\end{equation}
	Though $r_{0,j}\in F_2[[x]]$, only the part with degree less than $L_{i_{k-1}}$ participates in the above calculation. Hence, we only need to consider the part of $r_{0,j}$ modulo $x^{L_{i_{k-1}}}$. Therefore, the whole space of the solution vector $(r_{0,0},r_{0,1},\cdots,r_{0,k-2})$ can be regarded as $\{F_2[x]/x^{L_{i_{k-1}}}\}^{k-1}$.
	The value of $P(\{Z^{(s_0)}_{i_m}=z_{i_m}\}_{m=1}^{k-1})$ is equal to the ratio of the number of solution $(r_{0,0},r_{0,1},\cdots,r_{0,k-2})$ of (\ref{key359}) in the whole space $\{F_2[x]/x^{L_{i_{k-1}}}\}^{k-1}$. 
	
	By similar analysis, the value of $P(\{Z^{(s_1)}_{i_m}=z_{i_m}\}_{m=1}^{k-1})$ is equal to the ratio of the number of solution of (\ref{key360}) in the whole space $\{F_2[x]/x^{L_{i_{k-1}}}\}^{k-1}$.
	\begin{equation}\label{key360}
		\begin{aligned}
			\left
			\{
			\begin{array}{l}
				z_{i_1}=\sum_{j=0}^{k-2}r_{1,j}{y_{i_1}}^j+s_1{y_{i_1}}^{k-1}\pmod {x^{L_{i_1}}},\\
				z_{i_2}=\sum_{j=0}^{k-2}r_{1,j}{y_{i_2}}^j+s_1{y_{i_2}}^{k-1}\pmod {x^{L_{i_2}}},\\
				\hspace{10em}\cdots\hspace{10em}\\
				z_{i_{k-1}}=\sum_{j=0}^{k-2}r_{1,j}{y_{i_{k-1}}}^j+s_1{y_{i_{k-1}}}^{k-1}\pmod {x^{L_{i_{k-1}}}}.\\
			\end{array} \right.
		\end{aligned}
	\end{equation}
	By Theorem~\ref{thm6}, the numbers of solutions of the congruence equations (\ref{key359}) and (\ref{key360}) are same, which implies $P(\{Z^{(s_0)}_{i_m}=z_{i_m}\}_{m=1}^{k-1})=P(\{Z^{(s_1)}_{i_m}=z_{i_m}\}_{m=1}^{k-1})$. Then the security of the proposed scheme is completely proved using Lemma~\ref{lem1}.
	
	Now, we will give an example to show how to distribute shares and to restore secret.
	\\
	\noindent\textbf{Example.} In this example, we choose $k=3$ and the secret $s=110$. For three participants $P_2$, $P_3$ and $P_4$, let $100$, $101$ and $11000$ be the prefix codes of $2$, $3$ and $4$, respectively. As $\ell+(k-1)(\ell_4-1)=11$, the algorithm $\mathcal{E}$ randomly chooses two $11$-bit binary strings $r_0=01001101000$ and $r_1=10011001001$. Then the share of $P_2$ is given by
	\begin{align}
		Z^{(s)}_2=&r_0+r_1y_2+s{y_2}^2\nonumber\\
		=&(x+x^4+x^5)+(1+x^3+x^4)+(1+x)\nonumber\\
		=&x^3+x^5 \pmod {x^7},
	\end{align}
	thus, the share $Z^{(s)}_2$ is $0001010$.
	
	We proceed to construct the share $Z^{(s)}_3$. According to the algorithm $\mathcal{E}$, the share of $P_3$ is given by
	\begin{align}
		Z^{(s)}_3=&r_0+r_1y_3+s{y_3}^2\nonumber\\
		=&(x+x^4+x^5)+(1+x^3+x^4)(1+x^2)+(1+x)(1+x^2)^2\nonumber\\
		=&x^2+x^3+x^4+x^5+x^6 \pmod {x^7},
	\end{align}
	thus, the share $Z^{(s)}_3$ is $0011111$.
	
	And the share of $P_4$ is calculated by
	\begin{align}
		Z^{(s)}_4=&r_0+r_1y_4+s{y_4}^2\nonumber\\
		=&(x+x^4+x^5+x^7)+(1+x^3+x^4+x^7+x^{10})(1+x)+(1+x)(1+x)^2\nonumber\\                =&x+x^2+x^4+x^8+x^{10} \pmod {x^{11}},
	\end{align}
	thus, the share $Z^{(s)}_4$ is $01101000101$.
	
	Next, we take this example to restore the secret $s$. As defined above, let $100$, $101$ and $11000$ be the prefix codes of $2$, $3$ and $4$, respectively. Let  $0001010$, $0011111$ and $01101000101$ be the shares of the three participants $P_2$, $P_3$, $P_4$, respectively. Then we have 
	\begin{equation}
		\begin{aligned}
			\left
			\{
			\begin{array}{l}
				Z^{(s)}_2=x^3+x^5=r_0+r_1+s \pmod {x^7},\\
				Z^{(s)}_3=x^2+x^3+x^4+x^5+x^6=r_0+r_1(1+x^2)+s(1+x^2)^2 \pmod {x^7},\\
				Z^{(s)}_4=x+x^2+x^4+x^8+x^{10}=r_0+r_1(1+x)+s(1+x)^2 \pmod {x^{11}}.\\
			\end{array} \right.
		\end{aligned}
	\end{equation}
	We calculate $\alpha=8$ and $(y_2-y_4)(y_2-y_3)(y_3-y_4)=x^4(x+1)$ in $F_2[x]$. Since the bit length of the secret $s$ is $3$, the algorithm $\mathcal{R}$ solves the following equation  
	\begin{align*}
		s\Big(\prod_{1\leq u<v \leq 3} (y_{i_u}-y_{i_v})\Big)=\sum_{m=1}^{3}\prod_{\substack{1\leq u<v \leq 3\\u,v\neq m}} (y_{i_u}-y_{i_v})Z^{(s)}_{i_m} \pmod {x^8}.
	\end{align*}
	with a unique solution, i.e.
	\begin{align*}
		s=&\frac{\sum_{m=1}^{3}\prod_{\substack{1\leq u<v \leq 3\\u,v\neq m}} (y_{i_u}-y_{i_v})Z^{(s)}_{i_m}/x^{\sum_{1\leq u<v \leq 3}L_{i_u,i_v}}}{\prod_{1\leq u<v \leq 3} (y_{i_u}-y_{i_v})/x^{\sum_{1\leq u<v \leq 3}L_{i_u,i_v}}}\pmod {x^3} \nonumber\\
		=&\frac{[Z^{(s)}_2(y_3-y_4)+Z^{(s)}_3(y_2-y_4)+Z^{(s)}_4(y_2-y_3)]/x^4}{(y_2-y_4)(y_2-y_3)(y_3-y_4)/x^4}\pmod {x^3} \nonumber\\
		=&\frac{(x^4+x^6+x^{10}+x^{12})/x^4}{(x^4+x^5)/x^4}\pmod {x^3} \nonumber\\ 
		=&\frac{1+x^2}{1+x}\pmod {x^3}\nonumber\\
		\overset{(d)}=&(1+x^2)(1+x+x^2)\pmod {x^3} \nonumber\\
		=&1+x \pmod {x^3},
	\end{align*}
	where (d) holds since $(1+x)(1+x+x^2)=1$ in $F_2[x]/x^3$, thus the secret is $110$, which is correct.
	
	\section{Construction of the Evolving $k$-threshold Secret Sharing Scheme on A Polynomial Ring $F_p[x]$}\label{sec5}
	As described in the prior sections, based on binary prefix coding, we have proposed the constructions of evolving $k$-threshold secret sharing scheme in $F_2[x]$, where $k\geq 2$. Based on $p$-ary prefix coding, we consider the secret $s\in\{0,1,\cdots,p-1\}^{\ell}$ for any $p\in \mathbb{N}^+$, and we extend the proposed evolving $k$-threshold secret sharing scheme to $F_p[x]$.
	
	\subsection{Proposed Scheme}
	Given a set of $p$-ary prefix codes for positive integers, the codeword of the integer $i$ is denoted as $c_i=(c_{i,0},c_{i,1},\cdots,c_{i,\ell_i-1})$, where $c_{i,j}\in F_p$ for $0\leq j\leq \ell_i-1$ and $\ell_i$ denotes the codeword length of $c_i$. The polynomial form of $c_i$ is defined as $y_i=\sum_{j=0}^{j=\ell_i-1}c_{i,j}x^j\in F_p[x]$. For any $i\in \mathbb{N}^+$, then the $i$-th share with the algorithm $\mathcal{E}$ is defined as
	\begin{equation}\label{key466}
		Z^{(s)}_i=\sum_{j=0}^{k-2}r_j{y_i}^j+s{y_i}^{k-1}\pmod {x^{(\ell_i-1)(k-1)+\ell}},
	\end{equation}
	where $r_j$ is randomly chosen from $F_p[[x]]$ for $0\leq j\leq k-2$. %The actual value situation of each $r_j$ has been shown in Subsection~\ref{sec4.1}.  
	
	Given any $k$ participants $P_{i_1}, P_{i_2},\cdots, P_{i_k}\in \mathcal{P}$ with $i_1<\cdots<i_{k}$, let $L_{i_u,i_v}$ denote the maximal integer of $y_{i_u}-y_{i_v}$ satisfying $x^{L_{i_u,i_v}}\mid (y_{i_u}-y_{i_v})$, for any $u,v\in[k]$ with  $u\neq v$. Let $\alpha=\min_{m=1}^{k}\{L_{i_m}+\sum_{\substack{1\leq u<v \leq k\\u,v\neq m}} L_{i_u,i_v}\}$, where $L_{i_m}=(\ell_{i_m}-1)(k-1)+\ell$ for $1\leq m\leq k$.
	The algorithm $\mathcal{R}$ finds out that the following equation
	\begin{equation}
		s\Big(\prod_{1\leq u<v \leq k} (y_{i_u}-y_{i_v}) \Big)=\sum_{m=1}^{k}(-1)^{m-1}\prod_{\substack{1\leq u<v \leq k\\u,v\neq m}} (y_{i_u}-y_{i_v})Z^{(s)}_{i_m}\pmod {x^\alpha}
	\end{equation}
	has a unique solution of $s$ in $F_p[x]/x^\ell$ as
	\begin{align}\label{key4711}
		s=\frac{\sum_{m=1}^{k}(-1)^{m-1}\prod_{\substack{1\leq u<v \leq k\\u,v\neq m}} (y_{i_u}-y_{i_v})Z^{(s)}_{i_m}/x^{\sum_{1\leq u<v \leq k}L_{i_u,i_v}}}{\prod_{1\leq u<v \leq k} (y_{i_u}-y_{i_v})/x^{\sum_{1\leq u<v \leq k}L_{i_u,i_v}}}\pmod {x^\ell}.	
	\end{align}
	
	\noindent\textbf{Example.} In this example, we choose $k=3$, $p=3$ the secret $s=2101$ with $\ell=4$. For the three participants $P_2$, $P_5$ and $P_8$, let $01$, $102$ and $112$ be the prefix codes of $2$, $5$ and $8$, respectively. As $\ell+(k-1)(\ell_8-1)=8$, the algorithm $\mathcal{E}$ randomly chooses two $8$-bit binary strings $r_0=01201200$ and $r_1=20100010$, then the share of $P_2$ is given by
	\begin{align*}
		Z^{(s)}_2=&r_0+r_1y_2+s{y_2}^2\nonumber\\
		=&(x+2x^2+x^4+2x^5)+(2+x^2)x+(2+x+x^3)x^2\nonumber\\
		=&x^2+2x^3+x^4 \pmod {x^6},
	\end{align*}
	thus, the share $Z^{(s)}_2$ is $001210$.
	
	The share of $P_5$ is calculated by
	\begin{align*}
		Z^{(s)}_5=&r_0+r_1y_5+s{y_5}^2\nonumber\\
		=&(x+2x^2+x^4+2x^5)+(2+x^2+x^6)(1+2x^2)+(2+x+x^3)(1+2x^2)^2\nonumber\\ 
		=&1+2x+2x^3+2x^4+x^5+x^6+x^7 \pmod {x^8},
	\end{align*}
	thus, the share $Z^{(s)}_5$ is $12022111$.
	
	And the share of $P_8$ is given by
	\begin{align*}
		Z^{(s)}_8=&r_0+r_1y_8+s{y_8}^2\nonumber\\
		=&(x+2x^2+x^4+2x^5)+(2+x^2+x^6)(1+x+2x^2)\nonumber\\
		&+(2+x+x^3)(1+x+2x^2)^2\nonumber\\
		=&1+2x+x^2+2x^4+2x^5+2x^6+2x^7 \pmod {x^8},
	\end{align*}
	thus, the share $Z^{(s)}_8$ is $12102222$.
	
	Let $01$, $102$ and $112$ be the prefix codes of $2$, $5$ and $8$, respectively. Let  $001210$,  $12022111$ and $12102222$ respectively be the shares of the three participants $P_2$, $P_5$, $P_8$. Then we have 
	\begin{equation*}
		\begin{aligned}
			\left
			\{
			\begin{array}{l}
				Z^{(s)}_{P_2}=x^2+2x^3+x^4 \pmod {x^6},\\
				Z^{(s)}_{P_5}=1+2x+2x^3+2x^4+x^5+x^6+x^7  \pmod {x^8},\\
				Z^{(s)}_{P_8}=1+2x+x^2+2x^4+2x^5+2x^6+2x^7\pmod {x^{8}}.\\
			\end{array} \right.
		\end{aligned}
	\end{equation*}
	As $(y_2-y_5)(y_2-y_8)(y_5-y_8)=x(2+x+2x^2+2x^3+2x^4)$ and the bit length of the secret $s$ is $4$, the algorithm $\mathcal{R}$ reconstructs the secret $s$ as
	\begin{align*}
		s=&\frac{\sum_{m=1}^{3}(-1)^{m-1}\prod_{\substack{1\leq u<v \leq 3\\u,v\neq m}} (y_{i_u}-y_{i_v})Z^{(s)}_{i_m}/x^{\sum_{1\leq u<v \leq 3}L_{i_u,i_v}}}{\prod_{1\leq u<v \leq 3} (y_{i_u}-y_{i_v})/x^{\sum_{1\leq u<v \leq 3}L_{i_u,i_v}}}\pmod {x^4}\nonumber\\
		=&\frac{[Z^{(s)}_2(y_5-y_8)-Z^{(s)}_5(y_2-y_8)+Z^{(s)}_8(y_2-y_5)]/x^{\sum_{1\leq u<v \leq 3}L_{i_u,i_v}}}{(y_2-y_5)(y_2-y_8)(y_5-y_8)/x^{\sum_{1\leq u<v \leq 3}L_{i_u,i_v}}}\pmod {x^4}\nonumber\\	
		=&\frac{x(1+x+2x^2+2x^3+x^4+x^5+2x^6+x^8)/x}{x(2+x+2x^2+2x^3+2x^4)/x}\pmod {x^4}\nonumber\\ 
		=&\frac{1+x+2x^2+2x^3}{2+x+2x^2+2x^3}\pmod {x^4}\nonumber\\
		=&(1+x+2x^2+2x^3)(2+x+2x^2+2x^3)^{-1} \pmod {x^4}\nonumber\\
		=&(1+x+2x^2+2x^3)(2+2x+2x^3)\pmod {x^4} \nonumber\\
		=&2+x+x^3 \pmod {x^4},
	\end{align*}
	thus the secret is $2101$.
	
	\subsection{Proofs of Correctness and Secrecy}
	We will prove the correctness and secrecy of the proposed scheme in this subsection.
	
	\noindent\textbf{The proof of Correctness.} For any $t\in \mathbb{N}^+, A \in\mathcal{A}_t$, we will prove that the secret $s$ can be correctly recovered by the shares of the participants in $A$. Similarly, we only need to consider the case of $|A|=k$ when $|A|\geq k$. 
	
	Using the similar proof method proposed in Subsection~\ref{sec4.2}, we just make slight modifications to the proof.  Denote the $k$ elements in $A$ as $P_{i_1}$, $P_{i_2}$, $\cdots$, $P_{i_k}$ with $i_1<\cdots<i_k\leq t$, the corresponding shares are as follows.
	\begin{equation}\label{key468}
		\begin{aligned}
			\left
			\{
			\begin{array}{c}
				Z^{(s)}_{i_1}=\sum_{j=0}^{k-2}r_j{y_{i_1}}^j+s{y_{i_1}}^{k-1}\pmod {x^{(\ell_{i_1}-1)(k-1)+\ell}},\\
				Z^{(s)}_{i_2}=\sum_{j=0}^{k-2}r_j{y_{i_2}}^j+s{y_{i_2}}^{k-1}\pmod {x^{(\ell_{i_2}-1)(k-1)+\ell}},\\
				\cdots\\
				Z^{(s)}_{i_k}=\sum_{j=0}^{k-2}r_j{y_{i_k}}^j+s{y_{i_k}}^{k-1}\pmod {x^{(\ell_{i_k}-1)(k-1)+\ell}}.\\
			\end{array} \right.
		\end{aligned}
	\end{equation} 
	
	Considering the congruence equations in (\ref{key468}), there exists $h_m(x)\in F_p[x]$ satisfying the equation 
	
	\begin{equation}\label{key469}
		\sum_{j=0}^{k-2}r_j{y_{i_m}}^j+s{y_{i_m}}^{k-1}-Z^{(s)}_{i_m}=h_m(x)x^{(\ell_{i_m}-1)(k-1)+\ell}.	
	\end{equation}
	Then, we multiply both sides of the equation (\ref{key469}) by $H_{i_m}(x)$, where $H_{i_m}(x)= (-1)^{m-1}{\prod_{\substack{1\leq u<v \leq k\\u,v\neq m}} (y_{i_u}-y_{i_v})}$. Performing the above steps for each congruent equation in (\ref{key468}), then, we have
	\begin{equation}\label{key4691}
		\begin{aligned}
			\left
			\{
			\begin{array}{c}
				H_{i_1}(x)\Big(\sum_{j=0}^{k-2}r_j{y_{i_1}}^j+s{y_{i_1}}^{k-1}-Z^{(s)}_{i_1}\Big)=H_{i_1}(x)h_1(x){x^{(\ell_{i_1}-1)(k-1)+\ell}},\\
				H_{i_2}(x)\Big(\sum_{j=0}^{k-2}r_j{y_{i_2}}^j+s{y_{i_2}}^{k-1}-Z^{(s)}_{i_2}\Big)
				=H_{i_2}(x)h_2(x)x^{(\ell_{i_2}-1)(k-1)+\ell},\\
				\cdots\\
				H_{i_k}(x)\Big(\sum_{j=0}^{k-2}r_j{y_{i_k}}^j+s{y_{i_k}}^{k-1}-Z^{(s)}_{i_k}\Big)=H_{i_k}(x)h_k(x)x^{(\ell_{i_k}-1)(k-1)+\ell}.\\
			\end{array} \right.
		\end{aligned}	
	\end{equation}
	Summing these equations, we can further get
	\begin{align}\label{key4692}
		&\sum_{m=1}^{k}H_{i_m}(x){y_{i_m}}^{k-1}s+\sum_{j=0}^{k-2} \sum_{m=1}^{k}H_{i_m}(x){y_{i_m}}^{j}r_j\nonumber\\
		=& \sum_{m=1}^{k}H_{i_m}(x)Z^{(s)}_{i_m}+\sum_{m=1}^{k}H_{i_m}(x)h_m(x)x^{(\ell_{i_m}-1)(k-1)+\ell}.
	\end{align}
	Consider the coefficient of $s$ as $\sum_{m=1}^{k}H_{i_m}(x){y_{i_m}}^{k-1}$ in $F_p[x]$, which is equal to the following Vandermonde determinant, i.e.
	\begin{align}\label{key4693}
		\sum_{m=1}^{k}H_{i_m}(x){y_{i_m}}^{k-1}
		=\left |\begin{array}{ccccc}
			1 &1   &\cdots &1\\
			y_{i_1} &y_{i_2} &\cdots &y_{i_{k}}\\
			\vdots &\vdots &\ddots &\vdots\\
			{y_{i_1}}^{k-1} & {y_{i_2}}^{k-1} &\cdots &{y_{i_{k}}}^{k-1} \\
		\end{array}\right|=\prod_{1\leq u<v \leq k} (y_{i_u}-y_{i_v}).
	\end{align}
	
	For arbitrary $j$ with $0\leq j\leq k-2$, the coefficient of $r_j$ is $\sum_{m=1}^{k}H_{i_m}(x){y_{i_m}}^{j}$ in $F_p[x]$, then we have
	\begin{equation}\label{key4694}
		\sum_{m=1}^{k}H_{i_m}(x){y_{i_m}}^{j}
		=\left |\begin{array}{ccccc}
			1 &1   &\cdots &1\\
			y_{i_1} &y_{i_2} &\cdots &y_{i_{k}}\\
			\vdots &\vdots &\ddots &\vdots\\
			{y_{i_1}}^{k-2} & {y_{i_2}}^{k-2} &\cdots &{y_{i_{k}}}^{k-2} \\
			{y_{i_1}}^{j} & {y_{i_2}}^{j} &\cdots &{y_{i_{k}}}^{j} \\
		\end{array}\right|=0.
	\end{equation} 
	Taking the results of \eqref{key4693} and \eqref{key4694} into \eqref{key4692},  we can further get
	\begin{align}\label{key4695}
		\prod_{1\leq u<v \leq k} (y_{i_u}-y_{i_v})s=\sum_{m=1}^{k}H_{i_m}(x)Z^{(s)}_{i_m}
		+\sum_{m=1}^{k}H_{i_m}(x)h_m(x)x^{(\ell_{i_m}-1)(k-1)+\ell}.
	\end{align}
	As $\alpha=\min_{m=1}^{k}\{L_{i_m}+\sum_{\substack{1\leq u<v \leq k\\u,v\neq m}} L_{i_u,i_v}\}$, where $L_{i_m}=(\ell_{i_m}-1)(k-1)+\ell$ for $1\leq m\leq k$, then 
	that each polynomial
	\begin{equation*}
		H_{i_m}(x)h_m(x)x^{L_{i_m}}={\prod_{\substack{1\leq u<v \leq k\\u,v\neq m}}}(y_{i_u}-y_{i_v})x^{L_{i_m}}h_m(x)
	\end{equation*}
	has the factor $x^\alpha$ for any $m$. Therefore, \eqref{key4695} can be derived into
	\begin{equation}\label{key470}
		\prod_{1\leq u<v \leq k} (y_{i_u}-y_{i_v})s\equiv \sum_{m=1}^{k}(-1)^{m-1}\prod_{\substack{1\leq u<v \leq k\\u,v\neq m}} (y_{i_u}-y_{i_v})Z^{(s)}_{i_m}\pmod {x^{\alpha}}. 
	\end{equation}
	
	As $\prod_{1\leq u<v \leq k} (y_{i_u}-y_{i_v})$ has the maximal integer $\sum_{1\leq u<v \leq k}L_{i_u,i_v}$, and the bit length of $s$ is $\ell$, we get $\sum_{1\leq u<v \leq k}L_{i_u,i_v}+\ell\leq \alpha$ (the proof is proposed in Subsection~\ref{sec4.2}). Using Theorem~\ref{thm3}, we can reconstruct the unique solution of $s$ in $F_p[x]/x^\ell$, which can be written as
	\begin{align}\label{key471}
		s=\frac{\sum_{m=1}^{k}(-1)^{m-1}\prod_{\substack{1\leq u<v \leq k\\u,v\neq m}} (y_{i_u}-y_{i_v})Z^{(s)}_{i_m}/x^{\sum_{1\leq u<v \leq k}L_{i_u,i_v}}}{\prod_{1\leq u<v \leq k} (y_{i_u}-y_{i_v})/x^{\sum_{1\leq u<v \leq k}L_{i_u,i_v}}}\pmod {x^\ell}.	
	\end{align}
	
	As for the security of the proposed scheme, a similar proof method has been mentioned in Subsection~\ref{sec4.2}. Therefore, we will no longer describe it here.
	
	\subsection{Construction Based on Two $p$-ary Prefix Codes}
	For the proposed scheme over $F_p[x]$, the $t$-th participant's share $Z^{(s)}_t$ can be regarded as a polynomial of $x$ over $F_p$. It can also be regarded as a finite symbol over $F_p$. Therefore, we denote by $D(Z^{(s)}_t)$ the number of symbols for $Z^{(s)}_t$, then we have
	\begin{equation}
		D(Z^{(s)}_t)=(k-1)(\ell_t-1)+\ell,	
	\end{equation}
	where $\ell_t$ is the codeword length of using $p$-ary prefix code to encode positive integer $t$. If choosing different $p$-ary prefix codes, the corresponding $D(Z^{(s)}_t)$ may be different. 
	
	As described in Subsection~\ref{sub3.3}, we have shown two binary prefix coding named $\gamma$ code and $\delta$ code. Now, we introduce two $p$-ary prefix codes named $M_1$ code and $M_2$ code. $M_1$ code is represented by 
	\begin{equation*}
		M_1(t)=\overbrace{0,0,\cdots,0}^{\lfloor \log_p {t}\rfloor}[t]_p,
	\end{equation*} 
	where $[t]_p$ is the $p$-ary expression. Thus, we have 
	\begin{equation}
		L(M_1(t))=2\lfloor\log_p t\rfloor+1.
	\end{equation}
	Hence, if using $M_1$ code as the prefix code, we have
	\begin{equation}
		D(Z^{(s)}_t)=2(k-1)\lfloor\log_p t\rfloor+\ell.		
	\end{equation}
	$M_2$ code is given by
	\begin{equation*}
		M_2(t)=M_1(\lfloor \log_p {t}\rfloor+1)[t]_p.
	\end{equation*}
	Then the codeword length of $M_2(t)$ is given by
	\begin{equation}
		L(M_2(t))=\lfloor \log_p {t}\rfloor+2\lfloor \log_p (\lfloor \log_p {t}\rfloor+1)\rfloor+2.
	\end{equation}
	Therefore, if using $M_2$ code as the $p$-ary prefix code, we have
	\begin{equation}
		D(Z^{(s)}_t)=(k-1)(\lfloor \log_p {t}\rfloor)+2(k-1)(\lfloor \log_p (\lfloor \log_p {t}\rfloor+1)\rfloor)+(k-1)+\ell.		
	\end{equation}
	
	\begin{table}[t]
		\renewcommand\arraystretch{1}
		\centering
		\caption{Comparision of evolving $k$-threshold schemes based on $p$-ary prefix coding.}
		\setlength{\tabcolsep}{0.9mm}
		\begin{tabular}{|c|c|c|}
			\hline
			threshold &algorithm &$D(Z^{(s)}_t)$\\\hline
			\multirow{2}*{$k=2$}
			&~\cite{okamura2020new}&$(\lfloor \log_p {t}\rfloor+2\lfloor \log_p (\lfloor \log_p {t}\rfloor+1)\rfloor+2)\cdot\max\{\lceil \lg {(p+1)}\rceil,\ell\}$\\
			\cline{2-3}
			~& \textbf{ours}& \bm {${\lfloor \log_p {t}\rfloor+2\lfloor \log_p (\lfloor \log_p {t}\rfloor+1)\rfloor+1+\ell}$}\\
			\hline
			\multirow{2}*{$k\geq 3$}
			&none&/\\
			\cline{2-3}
			~&\textbf{ours}&\bm {${(k-1)(\lfloor \log_p {t}\rfloor+2\lfloor \log_p (\lfloor \log_p {t}\rfloor+1)\rfloor+1)+\ell}$}\\
			\hline				
		\end{tabular}
		\label{tab3}
	\end{table} 

	We tabulate the currently known evolving threshold secret sharing schemes which are based on $p$-ary prefix coding, and compare the corresponding $D(Z^{(s)}_t)$. For the proposed scheme, we analyze the corresponding $D(Z^{(s)}_t)$ for using $M_2$ code as the $p$-ary prefix code. In Table~\ref{tab3}, we also represent the value of the lowest $D(Z^{(s)}_t)$ in bold for each case of $k$. When $k=2$, the proposed scheme's $D(Z^{(s)}_t)$ is lower than the scheme in \cite{okamura2020new}. When $k\geq3$, there are no other schemes based on $p$-ary prefix coding. However, the proposed construction is applicable to any $k$. 
	
	\section{Conclusion and Discussion}\label{sec6}
	In this paper, based on the prefix coding, we proposed the algebraic-oriented constructions of evolving $k$-threshold schemes for an $\ell$-bit secret over a polynomial ring. Specifically, we first proposed the evolving $2$-threshold secret sharing scheme on $F_2[x]$ for binary secret, and then we extended the scheme to the evolving $k$-threshold scheme on $F_2[x]$. Finally, we gave a construction of an evolving $k$-threshold scheme over a polynomial ring $F_p[x]$ considering the secret $s\in\{0,1,\cdots,p-1\}^{\ell}$. The proposed schemes can establish the connection between prefix codes and the evolving schemes for $k\geq2$, and also the first evolving $k$-threshold secret sharing scheme by generalizing Shamir's scheme onto a polynomial ring. Besides, when $k=2$, the proposed scheme can be unified to describe all known evolving $2$-threshold secret sharing schemes which are based on prefix codes. In addition, we show that the share size of the $t$-th share is $(k-1)(\ell_t-1)+\ell$, where $\ell_t$ denotes the codeword length of the integer $t$ encoded by the given binary prefix code. When $\delta$ code is applied, the $t$-th share size is $(k-1)\lfloor\lg t\rfloor+2(k-1)\lfloor\lg ({\lfloor\lg t\rfloor+1}) \rfloor+\ell$, which is smaller than Komargodski et al's scheme, when $k>3$ and $\ell>1$. 
	
	There are still some thought-provoking and challenging issues that remain unresolved, which are also our future works.
	
	1) The correctness and security of the proposed schemes are perfect. If relaxing correctness or security, is it possible to propose more interesting and efficient schemes to achieve better efficiency or lower share size?
	
	2) Compared with evolving $3$-threshold scheme given in~\cite{DARCO2021149}, the share size of the proposed scheme is larger. Whether the proposed scheme be further improved to achieve a smaller share size? %If not, check whether the scheme~\cite{d2021secret} be extended for other cases. 
	
	\section{Appendices}
	\textbf{Proof of Theorem 4.}
	We will prove the conclusion of Theorem 4 by contradiction. Without losing generality, we assume that the congruence equation (\ref{key356}) has $N_0$ solutions and the congruence equation (\ref{key357}) has $N_1$ solutions with $N_0> N_1$. Let $\mathbb{R}_0$ and $\mathbb{R}_1$ respectively be the solution spaces of the congruence equations (\ref{key356}) and (\ref{key357}). According to the definition of congruence equations in (\ref{key356}), there exists $w_{0,j}$ satisfying following equation  
	\begin{equation}\label{key584}
		\begin{aligned}
			\left
			\{
			\begin{array}{l}
				z_{i_1}=\sum_{j=0}^{k-2}r_{0,j}{y_{i_1}}^j+s_0{y_{i_1}}^{k-1}+w_{0,1}x^{L_{i_1}},\\
				z_{i_2}=\sum_{j=0}^{k-2}r_{0,j}{y_{i_2}}^j+s_0{y_{i_2}}^{k-1}+w_{0,2}x^{L_{i_2}},\\
				\hspace{10em}\cdots\hspace{10em}\\
				z_{i_{k-1}}=\sum_{j=0}^{k-2}r_{0,j}{y_{i_{k-1}}}^j+s_0{y_{i_{k-1}}}^{k-1}+w_{0,k-1}x^{L_{i_{k-1}}},\\
			\end{array} \right.
		\end{aligned}
	\end{equation}
	where $w_{0,m}\in F_p[x]/x^{(\ell_{i_{k-1}}-1)(2k-3)+\ell-L_{i_m}}$ for $1\leq m \leq k-1$.  
	
	Correspondingly, there exists $w_{1,m}\in F_p[x]/x^{(\ell_{i_{k-1}}-1)(2k-3)+\ell-L_{i_m}}$ to convert the	congruence equation (\ref{key357}) into the following form 
	\begin{equation}\label{key585}
		\begin{aligned}
			\left
			\{
			\begin{array}{l}
				z_{i_1}=\sum_{j=0}^{k-2}r_{1,j}{y_{i_1}}^j+s_1{y_{i_1}}^{k-1}+w_{1,1}x^{L_{i_1}},\\
				z_{i_2}=\sum_{j=0}^{k-2}r_{1,j}{y_{i_2}}^j+s_1{y_{i_2}}^{k-1}+w_{1,2}x^{L_{i_2}},\\
				\hspace{10em}\cdots\hspace{10em}\\
				z_{i_{k-1}}=\sum_{j=0}^{k-2}r_{1,j}{y_{i_{k-1}}}^j+s_1{y_{i_{k-1}}}^{k-1}+w_{1,k-1}x^{L_{i_{k-1}}}.\\
			\end{array} \right.
		\end{aligned}
	\end{equation} 
	The above equations (\ref{key584}) and (\ref{key585}) do not change the number of solutions of the congruence equations (\ref{key356}) and (\ref{key357}), respectively. It means that the number of solutions $(r_{0,0}, r_{0,1},\cdots, r_{0,k-2},w_{0,1}, w_{0,2},\cdots, w_{0,k-1})$ of (\ref{key584}) and the number of solutions of (\ref{key356}) are same, the numbers of solutions of (\ref{key585}) and (\ref{key357}) are same. Next, we directly analyze the solution cases of the equations (\ref{key584}) and (\ref{key585}). According to the assumptions, the congruence equation (\ref{key356}) has solutions. As $\mathbb{R}_0$ is the solution space of the congruence equation (\ref{key356}), we choose one of solutions in $\mathbb{R}_0$ and denote it as $R^1_0=(r^1_{0,0}, r^1_{0,1},\cdots, r^1_{0,k-2})$, then $R^1_{0,+}=(r^1_{0,0}, r^1_{0,1},\cdots, r^1_{0,k-2}, w_{0,1}, w_{0,2},\cdots, w_{0,k-1})$ is the solution of the equation (\ref{key584}), where $w_{0,m}=(\sum_{j=0}^{k-2}r^1_{0,j}{y_{i_m}}^j+s_0{y_{i_m}}^{k-1}-z_{i_m})/x^{L_{i_m}}$ for $1\leq m\leq k-1$.
	
	Consider the quotient ring $F_p[x]/x^{(\ell_{i_{k-1}}-1)(2k-3)+\ell}$. Subtracting the equation (\ref{key585}) from the equation (\ref{key584}), then we have
	\begin{equation}\label{key586}
		\begin{aligned}
			\left
			\{
			\begin{array}{l}
				(s_0-s_1){y_{i_1}}^{k-1}=\sum_{j=0}^{k-2}r_j{y_{i_1}}^j+w_1x^{L_{i_1}},\\
				(s_0-s_1){y_{i_2}}^{k-1}=\sum_{j=0}^{k-2}r_j{y_{i_2}}^j+w_2x^{L_{i_2}},\\
				\hspace{10em}\cdots\hspace{10em}\\
				(s_0-s_1){y_{i_{k-1}}}^{k-1}=\sum_{j=0}^{k-2}r_j{y_{i_{k-1}}}^j+w_{k-1}x^{L_{i_{k-1}}},\\
			\end{array} \right.
		\end{aligned}
	\end{equation}
	where $r_j=r_{1,j}-r_{0,j}$ for $0\leq j\leq k-2$ and $w_m=w_{1,m}-w_{0,m}$ for $1\leq m\leq k-1$.
	
	Suppose there exists a solution $R=(r_0,\cdots,r_{k-2},w_1,\cdots,w_{k-1})$ of the equation (\ref{key586}).
	Then combine $R$ with the solution $R^1_{0,+}$ of the equation (\ref{key584}), and we assert that $R^1_{0,+}+R$ is the solution of the equation (\ref{key585}). Let $R'=(r_0,\cdots,r_{k-2})$,
	then, $R^1_{0}+R'=(r^1_{0,0}+r_0, r^1_{0,1}+r_1,\cdots, r^1_{0,k-2}+r_{k-2})$ is a solution of the equation (\ref{key357}). When $R^1_{0}$ is taken across the whole solution space $\mathbb{R}_0$, the congruence
	equation (\ref{key357}) has at least $N_0$ different solutions, which contradicts that (\ref{key357}) has $N_1$ solutions.  Therefore, the equations (\ref{key356}) and (\ref{key357}) have the same number of solutions.  Next, our goal is to find the solution of the equation (\ref{key586}).
	
	We consider the linear equations
	\begin{equation}\label{key5861}
		\begin{aligned}
			\left
			\{
			\begin{array}{l}
				(s_0-s_1){y_{i_1}}^{k-1}=\sum_{j=0}^{k-2}r_j{y_{i_1}}^j,\\
				(s_0-s_1){y_{i_2}}^{k-1}=\sum_{j=0}^{k-2}r_j{y_{i_2}}^j,\\
				\hspace{10em}\cdots\hspace{10em}\\
				(s_0-s_1){y_{i_{k-1}}}^{k-1}=\sum_{j=0}^{k-2}r_j{y_{i_{k-1}}}^j,\\
			\end{array} \right.
		\end{aligned}
	\end{equation}
	which can be written as 
	\begin{equation}\label{key5682}
		\left (\begin{array}{ccccc}
			1 &y_{i_1}   &\cdots &{y_{i_{k}}}^{k-2}\\
			1 &y_{i_2} &\cdots &{y_{i_{2}}}^{k-2}\\
			\vdots &\vdots &\ddots &\vdots\\
			1 & y_{i_{k-1}} &\cdots &{y_{i_{k-1}}}^{k-2} \\
		\end{array}\right)
		\left (\begin{array}{cc}
			r_0\\
			r_1\\
			\vdots\\
			r_{k-2} \\
		\end{array}\right)=(s_0-s_1)\left (\begin{array}{cc}
			{y_{i_1}}^{k-1}\\
			{y_{i_2}}^{k-1}\\
			\vdots\\
			{y_{i_{k-1}}}^{k-1} \\
		\end{array}\right).
	\end{equation}
	Let $A=\left (\begin{array}{ccccc}
		1 &y_{i_1}   &\cdots &{y_{i_{k}}}^{k-2}\\
		1 &y_{i_2} &\cdots &{y_{i_{2}}}^{k-2}\\
		\vdots &\vdots &\ddots &\vdots\\
		1 & y_{i_{k-1}} &\cdots &{y_{i_{k-1}}}^{k-2} \\
	\end{array}\right)$, $r=\left (\begin{array}{cc}
		r_0\\
		r_1\\
		\vdots\\
		r_{k-2} \\
	\end{array}\right)$ and $b=(s_0-s_1)\left (\begin{array}{cc}
		{y_{i_1}}^{k-1}\\
		{y_{i_2}}^{k-1}\\
		\vdots\\
		{y_{i_{k-1}}}^{k-1} \\
	\end{array}\right)$.
	We observe that $A$ is Vandermonde matrix.  Since ${y_{i_m}}$ with $1\leq m\leq k-1$ is the polynomial form of the prefix code, then $rank(A)=rank(A,b)$. Hence,  (\ref{key5682}) has solutions. Call that a fact, for a monic polynomial $f(y)=y^{k-1}+\sum_{i=1}^{k-1}a_iy^{k-1-i}$ with known $k-1$ roots $(y_1,y_2,\cdots,y_{k-1})$, then the $k-1$ coefficients $(a_1,a_2,\cdots,a_{k-1})$ are uniquely determined. By the formulas of the relation between roots and coefficients, we can calculate
	$a_i=(-1)^i\sigma_i(y_1,y_2,\cdots,y_{k-1})=\sum_{\substack{C\subseteq[k-1]\\|C|=i}}\prod_{g\in C} y_{g}$. Therefore, we can obtain a solution about $r$ in (\ref{key5682}), i.e. 
	\begin{align}\label{key5683}
		r_j=&(-1)^{k-2-j}(s_0-s_1)\sigma_{k-1-j}({y_{i_1}},{y_{i_2}},\cdots,{y_{i_{k-1}}})
		\nonumber\\
		=&(-1)^{k-2-j}(s_0-s_1)\sum_{\substack{C\subseteq[k-1]\\|C|=k-1-j}}{\prod_{g\in C} y_{i_g}}	
	\end{align}
	for $0\leq j\leq k-2$.
	
	From above, we can construct a solution of the equation (\ref{key586}) as
	\begin{equation}\label{key587}
		R=(\{r_j=(-1)^{k-2-j}(s_0-s_1)\sum_{\substack{C\subseteq[k-1]\\|C|=k-1-j}}{\prod_{g\in C} y_{i_g}}\}_{j=0}^{k-2}, \{w_m=0\}_{m=1}^{k-1}).
	\end{equation}

\bibliographystyle{IEEEtran} 
\bibliography{secret}

\begin{thebibliography}{10}
\providecommand{\url}[1]{#1}
\csname url@samestyle\endcsname
\providecommand{\newblock}{\relax}
\providecommand{\bibinfo}[2]{#2}
\providecommand{\BIBentrySTDinterwordspacing}{\spaceskip=0pt\relax}
\providecommand{\BIBentryALTinterwordstretchfactor}{4}
\providecommand{\BIBentryALTinterwordspacing}{\spaceskip=\fontdimen2\font plus
\BIBentryALTinterwordstretchfactor\fontdimen3\font minus
  \fontdimen4\font\relax}
\providecommand{\BIBforeignlanguage}[2]{{%
\expandafter\ifx\csname l@#1\endcsname\relax
\typeout{** WARNING: IEEEtran.bst: No hyphenation pattern has been}%
\typeout{** loaded for the language `#1'. Using the pattern for}%
\typeout{** the default language instead.}%
\else
\language=\csname l@#1\endcsname
\fi
#2}}
\providecommand{\BIBdecl}{\relax}
\BIBdecl

\bibitem{shamir1979share}
A.~Shamir, ``How to share a secret,'' \emph{Communications of the ACM},
  vol.~22, no.~11, pp. 612--613, Nov. 1979.

\bibitem{blakley1979safeguarding}
G.~R. Blakley, ``Safeguarding cryptographic keys,'' in \emph{Managing
  Requirements Knowledge, International Workshop on}.\hskip 1em plus 0.5em
  minus 0.4em\relax IEEE Computer Society, 1979, pp. 313--313.

\bibitem{fuyou2014randomized}
M.~Fuyou, X.~Yan, W.~Xingfu, and M.~Badawy, ``Randomized component and its
  application to ($ t $, $ m $, $ n $)-group oriented secret sharing,''
  \emph{IEEE Transactions on Information Forensics and Security}, vol.~10,
  no.~5, pp. 889--899, 2014.

\bibitem{harn2016realizing}
L.~Harn, C.~Hsu, M.~Zhang, T.~He, and M.~Zhang, ``Realizing secret sharing with
  general access structure,'' \emph{Information Sciences}, vol. 367, pp.
  209--220, 2016.

\bibitem{harn2010strong}
L.~Harn and C.~Lin, ``Strong (n, t, n) verifiable secret sharing scheme,''
  \emph{Information Sciences}, vol. 180, no.~16, pp. 3059--3064, 2010.

\bibitem{pang2005new}
L.-J. Pang and Y.-M. Wang, ``A new (t, n) multi-secret sharing scheme based on
  shamir’s secret sharing,'' \emph{Applied Mathematics and Computation}, vol.
  167, no.~2, pp. 840--848, 2005.

\bibitem{yang2004t}
C.-C. Yang, T.-Y. Chang, and M.-S. Hwang, ``A (t, n) multi-secret sharing
  scheme,'' \emph{Applied Mathematics and Computation}, vol. 151, no.~2, pp.
  483--490, 2004.

\bibitem{komargodski2016share}
I.~Komargodski, M.~Naor, and E.~Yogev, ``How to share a secret, infinitely,''
  in \emph{Theory of Cryptography Conference}.\hskip 1em plus 0.5em minus
  0.4em\relax Springer, 2016, pp. 485--514.

\bibitem{komargodski2017share}
{I. Komargodski}, M.~Naor, and E.~Yogev, ``How to share a secret, infinitely,''
  \emph{IEEE Transactions on Information Theory}, vol.~64, no.~6, pp.
  4179--4190, Jun. 2017.

\bibitem{DARCO2021149}
P.~D'Arco, R.~{De Prisco}, and A.~{De Santis}, ``Secret sharing schemes for
  infinite sets of participants: A new design technique,'' \emph{Theoretical
  Computer Science}, vol. 859, pp. 149--161, 2021.

\bibitem{elias1975universal}
P.~Elias, ``Universal codeword sets and representations of the integers,''
  \emph{IEEE transactions on information theory}, vol.~21, no.~2, pp. 194--203,
  1975.

\bibitem{karnin1983secret}
E.~Karnin, J.~Greene, and M.~Hellman, ``On secret sharing systems,'' \emph{IEEE
  Transactions on Information Theory}, vol.~29, no.~1, pp. 35--41, 1983.

\bibitem{bogdanov2016threshold}
A.~Bogdanov, S.~Guo, and I.~Komargodski, ``Threshold secret sharing requires a
  linear size alphabet,'' in \emph{Theory of Cryptography: 14th International
  Conference, TCC 2016-B, Beijing, China, October 31-November 3, 2016,
  Proceedings, Part II 14}.\hskip 1em plus 0.5em minus 0.4em\relax Springer,
  2016, pp. 471--484.

\bibitem{mignotte1983share}
M.~Mignotte, ``How to share a secret,'' in \emph{Cryptography: Proceedings of
  the Workshop on Cryptography Burg Feuerstein, Germany, March 29--April 2,
  1982 1}.\hskip 1em plus 0.5em minus 0.4em\relax Springer, 1983, pp. 371--375.

\bibitem{asmuth1983modular}
C.~Asmuth and J.~Bloom, ``A modular approach to key safeguarding,'' \emph{IEEE
  transactions on information theory}, vol.~29, no.~2, pp. 208--210, 1983.

\bibitem{liu2015novel}
Y.~Liu, L.~Harn, and C.-C. Chang, ``A novel verifiable secret sharing mechanism
  using theory of numbers and a method for sharing secrets,''
  \emph{International Journal of Communication Systems}, vol.~28, no.~7, pp.
  1282--1292, 2015.

\bibitem{ning2018constructing}
Y.~Ning, F.~Miao, W.~Huang, K.~Meng, Y.~Xiong, and X.~Wang, ``Constructing
  ideal secret sharing schemes based on chinese remainder theorem,'' in
  \emph{International Conference on the Theory and Application of Cryptology
  and Information Security}.\hskip 1em plus 0.5em minus 0.4em\relax Springer,
  2018, pp. 310--331.

\bibitem{benaloh1990generalized}
J.~Benaloh and J.~Leichter, \emph{Generalized secret sharing and monotone
  functions}.\hskip 1em plus 0.5em minus 0.4em\relax Springer, 1990.

\bibitem{pedersen1991non}
T.~P. Pedersen, ``Non-interactive and information-theoretic secure verifiable
  secret sharing,'' in \emph{Annual international cryptology conference}.\hskip
  1em plus 0.5em minus 0.4em\relax Springer, 1991, pp. 129--140.

\bibitem{brickell1989some}
E.~F. Brickell, ``Some ideal secret sharing schemes,'' in \emph{Workshop on the
  Theory and Application of of Cryptographic Techniques}.\hskip 1em plus 0.5em
  minus 0.4em\relax Springer, 1989, pp. 468--475.

\bibitem{gennaro1998simplified}
R.~Gennaro, M.~O. Rabin, and T.~Rabin, ``Simplified vss and fast-track
  multiparty computations with applications to threshold cryptography,'' in
  \emph{Proceedings of the seventeenth annual ACM symposium on Principles of
  distributed computing}, 1998, pp. 101--111.

\bibitem{cramer2000general}
R.~Cramer, I.~Damg{\aa}rd, and U.~Maurer, ``General secure multi-party
  computation from any linear secret-sharing scheme,'' in \emph{International
  Conference on the Theory and Applications of Cryptographic Techniques}.\hskip
  1em plus 0.5em minus 0.4em\relax Springer, 2000, pp. 316--334.

\bibitem{chor1985verifiable}
B.~Chor, S.~Goldwasser, S.~Micali, and B.~Awerbuch, ``Verifiable secret sharing
  and achieving simultaneity in the presence of faults,'' in \emph{26th Annual
  Symposium on Foundations of Computer Science (sfcs 1985)}.\hskip 1em plus
  0.5em minus 0.4em\relax IEEE, 1985, pp. 383--395.

\bibitem{simmons1990really}
G.~J. Simmons, ``How to (really) share a secret,'' in \emph{Conference on the
  Theory and Application of Cryptography}.\hskip 1em plus 0.5em minus
  0.4em\relax Springer, 1990, pp. 390--448.

\bibitem{krawczyk1993secret}
H.~Krawczyk, ``Secret sharing made short,'' in \emph{Annual international
  cryptology conference}.\hskip 1em plus 0.5em minus 0.4em\relax Springer,
  1993, pp. 136--146.

\bibitem{ding2021communication}
J.~Ding, C.~Lin, H.~Wang, and C.~Xing, ``Communication efficient secret sharing
  with small share size,'' \emph{IEEE Transactions on Information Theory},
  vol.~68, no.~1, pp. 659--669, 2021.

\bibitem{csirmaz2012line}
L.~Csirmaz and G.~Tardos, ``On-line secret sharing,'' \emph{Designs, Codes and
  Cryptography}, vol.~63, pp. 127--147, 2012.

\bibitem{komargodski2017evolving}
I.~Komargodski and A.~Paskin-Cherniavsky, ``Evolving secret sharing: dynamic
  thresholds and robustness,'' in \emph{Theory of Cryptography: 15th
  International Conference, TCC 2017, Baltimore, MD, USA, November 12-15, 2017,
  Proceedings, Part II 15}.\hskip 1em plus 0.5em minus 0.4em\relax Springer,
  2017, pp. 379--393.

\bibitem{d2018equivalence}
P.~D’Arco, R.~D. Prisco, and A.~D. Santis, ``On the equivalence of
  2-threshold secret sharing schemes and prefix codes,'' in \emph{International
  Symposium on Cyberspace Safety and Security}.\hskip 1em plus 0.5em minus
  0.4em\relax Springer, 2018, pp. 157--167.

\bibitem{okamura2020new}
R.~Okamura and H.~Koga, ``New constructions of an evolving 2-threshold scheme
  based on binary or d-ary prefix codes,'' in \emph{2020 International
  Symposium on Information Theory and Its Applications (ISITA)}.\hskip 1em plus
  0.5em minus 0.4em\relax IEEE, 2020, pp. 432--436.

\end{thebibliography}
\end{document}